\documentclass{article} 

\usepackage{float}      
\usepackage{placeins}   
\usepackage{graphicx}   

\usepackage[numbers,sort&compress]{natbib} 

\usepackage{amsthm}

\usepackage{latexsym}   
\usepackage{xcolor}     
\usepackage{lipsum}     
\usepackage{amsfonts}   
\usepackage{graphicx}
\usepackage{subcaption}
\usepackage{amssymb}    
\usepackage{amsmath}    
\usepackage{bbm}        
\usepackage{eucal}      
\usepackage{mathabx}    
\usepackage{mathrsfs}   
\usepackage{mathtools}
\usepackage{txfonts}

\usepackage{multirow}   
\usepackage{array}      
\usepackage{booktabs}   
\usepackage{ulem}       

\usepackage{todonotes}  

\usepackage{enumerate}  
\usepackage{enumitem}   

\newcommand{\githubrepo}{\url{https://github.com/tosmeow/passive-impact.git}}

\usepackage{thmtools}  
\usepackage[ruled,vlined]{algorithm2e} 
\usepackage{tikz-cd}

\newtheorem{theorem}{Theorem}[section]
\newtheorem{definition}[theorem]{Definition}

\newtheorem{proposition}[theorem]{Proposition}
\newtheorem{lemma}[theorem]{Lemma}
\newtheorem{example}[theorem]{Example}
\newtheorem{assumption}{Assumption}
\newtheorem{corollary}[theorem]{Corollary}
\newtheorem{remark}[theorem]{Remark}

\renewcommand{\theassumption}{\Alph{assumption}}

\makeatletter
\newcommand{\settheoremtag}[1]{
\let\oldtheassumption\theassumption
\renewcommand{\theassumption}{#1}
\g@addto@macro\endassumption{
\global\let\theassumption\oldtheassumption}
}
\makeatother

\def\wb{\widebar}

\newcommand{\1}{\mathbbm{1}} 

\DeclareMathOperator{\Poiss}{Poiss} 
\DeclareMathOperator{\Poisson}{Poisson} 

\newcommand{\norm}[1]{{\vert\kern-0.25ex\vert #1 \vert\kern-0.25ex\vert}}
\newcommand{\bignorm}[1]{{\big\vert\kern-0.25ex\big\vert #1 \big\vert\kern-0.25ex\big\vert}}
\newcommand{\opnorm}[1]{{\vert\kern-0.25ex\vert\kern-0.25ex\vert #1 \vert\kern-0.25ex\vert\kern-0.25ex\vert}}

\newcommand{\dd}{\mathrm{d}}

\newcommand{\E}{\mathbb{E}}

\newcommand{\lp}{\left\{}
\newcommand{\rp}{\right\}}
\newcommand{\R}{\mathbb{R}}

\newcommand{\ind}{\mathrm{1}}

\def\bbR{\mathbb{R}}

\def\bbN{\mathbb{N}}

\def\bbZ{\mathbb{Z}}

 \newcommand{\calD}{\mathcal{D}}
 \newcommand{\calF}{\mathcal{F}}
\newcommand{\calG}{\mathcal{G}}

 \newcommand{\calP}{\mathcal{P}}

\usepackage{hyperref}

\title
{Exact conditional simulation of Point processes: \\Application to pathwise market impact estimation}

\author
{Joseph Leclère\footnote{Ceremade, Université Paris Dauphine-PSL, \texttt{leclere@ceremade.dauphine.fr}} \and Youssef Ouazzani Chahdi\footnote{MICS, CentraleSupélec, \texttt{youssef.ouazzani-chahdi@centralesupelec.fr}} \and Mathieu Rosenbaum\footnote{Ceremade, Université Paris Dauphine-PSL, \texttt{mathieu.rosenbaum@dauphine.psl.eu}} \and Gr\'egoire Szymanski\footnote{DMATH, Université du Luxembourg, \texttt{gregoire.szymanski@uni.lu}}}

\begin{document}
\raggedbottom

\date{\today}

\maketitle

\begin{abstract}
    Market impact is defined as the difference between the observed price trajectory under a given execution strategy and the counterfactual trajectory that would have prevailed without it. Since this counterfactual is unobservable, estimating market impact requires simulating alternative paths under the same realized market randomness. We address this by studying the conditional simulation of point processes under perturbed intensities. 
    Given an observed counting process whose intensity is determined by its own history, we characterize the conditional law of the latent Poisson random measure in a thinning representation. This yields an exact, event-driven algorithm that reconstructs counterfactual paths on a common randomness source, enabling rigorous pathwise market impact estimation for aggressive, passive, and mixed strategies.
\end{abstract}

\maketitle

% ...



\section{Introduction}
\label{sec:introduction}

A recurring problem in stochastic modeling is to observe a system and ask what the same system would have done, had a control been applied to it. The object is not merely another independent realization of the model, but a counterfactual trajectory driven by the same primitive randomness as the observed controlled trajectory \citep{holland1986statistics,glasserman2004monte}. In favorable situations, the model structure and the observation determine this primitive randomness, or at least the part of it relevant for the counterfactual dynamics, and exact pathwise reconstruction is possible. In other situations, the observation is intrinsically coarser than the randomness that drives the system, so that the uncontrolled state cannot be identified as a single path from the controlled observation alone. \\

The obstruction is informational. A trajectory may reveal the jumps of a point process, or the state variables generated by those jumps, without revealing the auxiliary marks and rejected candidate events that were present in the underlying construction. More generally, if the $\sigma$-field generated by the observation is strictly smaller than the $\sigma$-field generated by the primitive noise, then replacing the missing noise by an arbitrary representative one would lead to some loss of information. The mathematically well-defined object is instead the conditional law of the counterfactual state given the factual state.
The construction below implements this program for point process models motivated by market impact.

\subsection{Market impact as a counterfactual problem}

Market impact is a central example of this counterfactual problem.
Transactions, order insertions and cancellations change queues, order flow intensities, and ultimately transaction prices. Nevertheless, execution, impact, and transaction cost analysis require comparison with the price process that would have prevailed without the trader's intervention, under the same realized market environment. If $P^{\mathrm{int}}$ denotes the price generated under the intervention and $P^0$ the corresponding reference price, the pathwise impact at time $t$ is
\begin{equation*}
\mathrm{MI}_t=P^{\mathrm{int}}_t-P^0_t.
\end{equation*}
The subtraction is meaningful only after specifying a coupling of the two price processes. Independent resimulation of $P^0$ would mix execution effects with unrelated market noise.\\

The literature provides several complementary notions of impact. Classical models of price formation and execution costs relate trades to permanent or temporary price changes \citep{kyle1985continuous,hasbrouck1991measuring,almgren2001optimal,webster2023handbook}. Empirical and propagator-based models emphasize signed, concave, transient responses under persistent order flow and no-arbitrage restrictions \citep{bouchaud2003fluctuations,bouchaud2009markets,jaisson2015market,jusselin2020noarbitrage,gatheral2010no,bacry2015hawkes}. These approaches justify studying average response curves and structural price functionals. For ex post cost analysis, however, two metaorders with the same size and duration may occur under different queue states, spreads, and contemporaneous order flow shocks. The quantity needed here is therefore not only an unconditional map from volume to average response, but the conditional law of market impact given the realized market environment. \\


The distinction is particularly acute for passive execution. A market order directly consumes displayed liquidity and enters signed market-order flow. A limit order supplies liquidity, may or may not be filled, and changes the future queue and market order dynamics through the state of the book. Queue reactive models make this feedback explicit by allowing event intensities to depend on the current queue state \citep{huang2015simulating}. In the passive impact model of \citep{youssef2024passiveimpact}, the price response of a market order is weighted by a queue-dependent coefficient. The price reflects the current anticipation of order flow imbalance and is modeled in a general form as follows: with aggregate ask and bid queues $q^a,q^b$ and market order counts $N^a,N^b$, the price used in Section~\ref{sec:impact} is
\begin{equation*}
P_t=P_0+\lim_{T\to\infty}
\E\!\left[
\int_0^T \kappa(q_s^a)\,\dd N_s^a
-
\int_0^T \kappa(q_s^b)\,\dd N_s^b
\ \middle|\ \calG_t
\right],
\end{equation*}
where $\kappa$ is non-increasing in the displayed depth and $(\calG_t)_t$ a filtration to be defined later. Section~\ref{sec:impact} recalls the assumptions under which this expression is finite, represents $P_t$ as a non-anticipative functional of the order book paths, and defines impact by evaluating the same functional on factual and counterfactual trajectories coupled through common primitive noise.

\subsection{Latent Poisson noise and conditional reconstruction}

Order books evolve through discrete events in continuous time: limit orders, cancellations, and market orders. Counting processes and marked point processes are therefore the natural language for the order flow \citep{bremaud1981point,daley2003introduction}. In the notation used below,
\begin{equation*}
X=(L^a,C^a,N^a,L^b,C^b,N^b)
\end{equation*}
collects these six event counts. If $X^i$ has predictable intensity $\lambda^i$, then, heuristically,
\begin{equation*}
\mathbb P\!\left(X^i_{t+h}-X^i_t=1\mid\calG_{t-}\right)
=\lambda^i_t h+o(h),
\qquad h\downarrow0,
\end{equation*}
with simultaneous jumps excluded by the standing assumptions. Dependence of $\lambda^i$ on past events and queue states expresses market endogeneity. Hawkes processes model clustering and persistence in order flow \citep{hawkes1971spectra,bacry2015hawkes}, while queue-reactive specifications model the feedback from displayed depth to event rates \citep{huang2015simulating}. \\

The coupling used in this paper is based on the Poisson random measure thinning representation. For a counting process $N$ with predictable intensity $\lambda$, we have
\begin{equation*}
N_t=\int_0^t\int_{\R_+}\1_{\{z\leq \lambda_s\}}\,\pi(\dd s,\dd z),
\end{equation*}
where $\pi$ is a Poisson random measure on $\R_+\times\R_+$. Once $\pi$ is fixed, a perturbation changes the acceptance region $\{z\le \lambda_s\}$ but not the candidate atoms. This is the common-randomness principle underlying exact thinning algorithms \citep{lewis1979simulation,ogata1981lewis} and the shared-noise comparison developed here. \\

The difficulty is precisely the non identifiability described above. The observed event path reveals the accepted jump times and types, hence the thinned counting process and the induced queue path, but it does not reveal the accepted vertical marks or the rejected atoms above the intensity graph. Counterfactual simulation from an observed trajectory therefore requires the conditional law of the latent Poisson measure given the thinned path. For history dependent intensities this is not a direct deterministic splitting argument, because the acceptance region is itself random and generated by the solution. Under the predictability, strong well-posedness and finite-horizon boundedness conditions stated in Section~\ref{sec:cond_distribution}, Theorem~\ref{thm:conditionallaw} gives the required decomposition: conditionally on the observed path, the accepted marks are independent uniforms below the corresponding intensity levels, while the unrevealed cloud on the complement of the corresponding accepted region remains an independent Poisson random measure with the restricted intensity.

\subsection{Contributions and organization}

The paper develops an exact conditional simulation framework for counterfactual market impact evaluation in point process based order book models. The first contribution is probabilistic. Theorem~\ref{thm:conditionallaw} identifies the conditional law of the Poisson random measure driving a strongly well-posed, non explosive thinned process with predictable path-dependent intensity and the finite-horizon bound of Subsection~\ref{sec:cond_distribution}. It extends the usual Poisson splitting intuition to the random acceptance regions generated by a strong thinning equation; the result reconstructs a conditional distribution of the latent noise.\\

The second contribution is algorithmic. Proposition~\ref{prop:conditional_simulation}, under the corresponding one-sided queue assumptions, turns this conditional law into an event-driven simulator: observed jumps are kept at their times with resampled admissible marks, the residual Poisson cloud is sampled on the unrevealed region, and the perturbed dynamics are solved by thinning on the reconstructed noise. No time discretization of the primitive point process is required.\\ 

The third contribution is the link with pathwise market impact. Section~\ref{sec:impact} formulates passive and aggressive impact as differences of non-anticipative price functionals evaluated on factual and counterfactual trajectories of the queue and order flow processes, coupled by the same Poisson point measures.  Section~\ref{sec:quantifying} applies the conditional simulator to these functionals and derives the tractable representation of Theorem~\ref{thm:explicit_impact_shape} under affine queue drift and exponential-sum Hawkes kernels. Section~\ref{sec:a_posteriori} treats the post-trade inverse problem: the impacted trajectory is observed, and the baseline trajectory that would have prevailed without the intervention is conditionally reconstructed under the same latent noise. For reproducibility, the code used to generate the numerical experiments is available at \githubrepo.\\

The rest of the paper is organized as follows. Section~\ref{sec:impact} fixes the pathwise market impact object and the shared-noise coupling. Section~\ref{sec:primer} recalls Poisson random measures, thinning, and common-noise perturbations, then states the conditional-law theorem. Section~\ref{sec:quantifying} gives the conditional simulation algorithm and its impact applications. Section~\ref{sec:a_posteriori} develops ex post impact and cost analysis. The appendices contain the proofs of the functional representation, the conditional law theorem, the simulation propositions, and the first-order impact formula.

\subsection{Notation and path-space conventions}
\label{subsec:intro_notation}

We write $\Poiss(r)$ for the Poisson law with parameter $r>0$. The symbol $\calD$ denotes the space of càdlàg functions from $\R_+$ to $\R$, equipped with the Kolmogorov $\sigma$-algebra generated by coordinate maps. For $x\in\calD$ and $t\ge0$, the stopping operator is
\begin{equation*}
S_t(x):=x_{\cdot\wedge t}.
\end{equation*}

\section{Pathwise market-impact theory under Poisson point process coupling}
\label{sec:impact}

This section fixes the market-impact object used in the sequel. Prices are first represented as non-anticipative functionals of queue and market order paths, following the anticipation of order flow viewpoint \citep{jaisson2015market, jusselin2020noarbitrage} and its queue-dependent extension in \citep{youssef2024passiveimpact}. Market impact is then the difference of this functional on factual and counterfactual trajectories coupled through the same underlying randomness. The final subsection identifies the latent noise reconstruction problem solved in the rest of the paper.
All probabilistic statements used below are inherited either from the Hawkes and queue reactive specifications recalled here or from the conditional simulation results proved in later sections.

\subsection{Price as anticipated queue-weighted order flow imbalance}
\label{sec:modeling_prices}

Under no-arbitrage assumption and imposing a linear price impact specification, the constant-sensitivity model of \citet{jaisson2015market} expresses the price as the conditional expectation of future signed market order imbalance. More precisely, we have
\begin{equation}
\label{eq:price_ofi_constant_kappa}
P_t = P_0 + \lim_{T\to\infty}\kappa\,\E\!\left[N^a_T-N^b_T\mid\calG_t\right],
\end{equation}
where $\kappa$ is the constant contribution attached to one signed market order. This parameter can be interpreted either as permanent impact per trade or as the average informational content of a single trade. In this baseline model, $\kappa$ is independent of the contemporaneous state of the book: it assigns the same informational price response to a trade regardless of the displayed depth it consumes. This assumption is quite restrictive as it is well known that high-frequency future price returns are highly linked to the current available liquidity. This limitation was addressed in \citet{youssef2024passiveimpact} where a queue-dependent coefficient is introduced to distinguish market orders hitting a thin queue from market orders hitting a deep queue. More precisely, we denote by $q^a$ and $q^b$ the displayed aggregated ask and bid queues and we define the queue-weighted price by
\begin{equation}
\label{eq:price}
P_t=P_0+\lim_{T\to\infty}\E\!\left[\int_0^T \kappa(q^a_s)\,\dd N^a_s-\int_0^T \kappa(q^b_s)\,\dd N^b_s\ \middle|\ \calG_t\right].
\end{equation}
Here $\kappa(q)$ is the contribution of a market order when the queue on the consumed side has size $q$. The standing monotonicity convention is that $\kappa$ is non-increasing: larger displayed depth attenuates the price response to aggressive orders.\\

The conditional expectation in Equations \eqref{eq:price_ofi_constant_kappa} and \eqref{eq:price} is therefore model dependent. We write it under a filtration $(\calG_t)_{t \geq 0}$ with respect to which the processes are progressively measurable. The Doob--Dynkin lemma \citep[see, e.g.,][]{MPS} yields a measurable map $\calP_t:\calD^4\to\R$ such that
\begin{equation*}
P_t=\calP_t\!\left(S_t(q^a),S_t(q^b),S_t(N^a),S_t(N^b)\right)\qquad\text{a.s.,}
\end{equation*}
and hence $P_t=\calP_t(q^a,q^b,N^a,N^b)$ a.s.\ for some non-anticipative $\calP_t:\calD^4\to\R$.
Here non-anticipative means that, for any $x,y\in\calD^4$,
\begin{equation*}
x_{\cdot\wedge t} = y_{\cdot\wedge t}\quad\Longrightarrow\quad \calP_t(x)=\calP_t(y),
\end{equation*}
with $S_t$ the stopping operator introduced in Subsection~\ref{subsec:intro_notation}. \\

This observation separates the pricing functional from the structural model used to generate paths. A reduced-form specification may estimate $\calP_t$ directly from data, whereas the present paper fixes an explicit point process model because counterfactual impact requires a coupling of factual and intervened trajectories.

\subsection{Hawkes and queue reactive specification}

We now recall the structural specification used in the sequel. It combines independent Hawkes market order flows \citep{hawkes1971spectra,bacry2015hawkes} with queue reactive limit order and cancellation intensities \citep{huang2015simulating}, in the form analyzed in \citep{youssef2024passiveimpact}.

The market order processes $N^a$ and $N^b$ are independent Hawkes processes with common baseline $\mu > 0$ and non-negative kernel $\varphi$ satisfying $\|\varphi\|_{L^1}<1$, such that their intensities are given by
\begin{equation*}
\lambda^x_t=\mu+\int_0^{t-}\varphi(t-s)\,\dd N^x_s,\qquad x\in\{a,b\}.
\end{equation*}

The corresponding queue processes are
\begin{equation}
\label{eq:queue}
q^a_t=q^a_0+L^a_t-C^a_t-N^a_t,\qquad
q^b_t=q^b_0+L^b_t-C^b_t-N^b_t,
\end{equation}
where, for $x\in\{a,b\}$,
\begin{equation*}
\lambda_t^{L,x}=\lambda^L(q^x_{t-}),\qquad
\lambda_t^{C,x}=\lambda^C(q^x_{t-}),
\end{equation*}
for some functions $\lambda^L, \lambda^C: \mathbb{R} \to \mathbb{R}_+$. 

As in \citep{youssef2024passiveimpact}, the state space is taken to be $\bbZ$. Negative values may occur in this stylized aggregated-queue model; the monotonicity and drift conditions below are imposed on $\bbZ$ precisely to keep the dynamics controlled without adding a reflecting boundary.\\

The existence of the limit in Equation \eqref{eq:price} is not obvious, because future Hawkes arrivals and future queue states remain correlated through the queue equations. We use the following sufficient conditions from \citep[Theorem~2.1]{youssef2024passiveimpact}.

\begin{assumption}
\label{assumption:passive}
Assume $\lambda^L$ is decreasing, $\lambda^C$ is increasing, and $\kappa$ is non-negative, bounded, and non-increasing. Also $\|\varphi\|_{L^1}<1$ and
\begin{equation}
\label{eq:uniform_drift_condition}
\mathfrak{m}_k:=\inf_{q\in\bbZ}\!\left\{\lambda^L(q)-\lambda^L(q+k)+\lambda^C(q+k)-\lambda^C(q)\right\},\qquad
\underline{\mathfrak{m}}:=\inf_{k\ge 1}\mathfrak{m}_k>0.
\end{equation}
\end{assumption}

Under Assumption~\ref{assumption:passive}, the limit in Equation\eqref{eq:price} is finite almost surely for every $t$. The assumptions are not claimed to be minimal; their role is to combine Hawkes stability with sufficient queue mean reversion. Variants in which $\kappa$ depends on both queues can be treated analogously when the corresponding integrability holds. In particular, $\|\varphi\|_{L^1}<1$ is the standard Hawkes stability condition, while Equation \eqref{eq:uniform_drift_condition} enforces a uniform restoring force for two coupled queues separated by $k$ units. These are the only well-posedness assumptions on the price model used in the rest of this section. \\

The preceding Doob--Dynkin argument gives a generic non-anticipative representation. The Markovian restart structure of the Hawkes/queue-reactive model gives the sharper decomposition needed below.

\begin{proposition}
\label{proposition:market_impact_functional}
Under Assumption~\ref{assumption:passive}, for every $t\ge 0$ there exists a measurable map
\begin{equation*}
\mathcal R_t:\R^2\times\mathcal D^2\longrightarrow\R
\end{equation*}
such that a.s.
\begin{equation}
\label{eq:functional_price_decomposition}
P_t = P_0 + \int_0^t \kappa(q_s^a)\,\dd N_s^a - \int_0^t \kappa(q_s^b)\,\dd N_s^b + \mathcal R_t(q_t^a,q_t^b,S_t(N^a),S_t(N^b)).
\end{equation}

\end{proposition}





The proof is given in Appendix~\ref{appendix:section_two_one}. Proposition~\ref{proposition:market_impact_functional} is the structural input for the rest of the section: at a fixed time, the price is a deterministic non-anticipative functional of the observed state variables.
The term $\mathcal R_t$ is the continuation value of future queue-weighted market order imbalance, conditional on the state and Hawkes prehistory available at time $t$.\\

When $\kappa$ is constant, \cite[Proposition~3.2]{jaisson2015market} gives the closed-form propagator representation
\begin{equation}
\label{eq:price_hawkes_propagator}
P_t=P_0+\kappa\int_0^t \xi(t-s)\,\dd\!\left(N^a_s-N^b_s\right),
\qquad
\xi(u)=1+\Big(1+\int_0^\infty \psi(r)\,\dd r\Big)\int_u^\infty \varphi(r)\,\dd r,
\end{equation}
where $\psi=\sum_{n\ge 1}\varphi^{*n}$. This propagator form will be used below only as a closed-form benchmark.
Since $\xi(u)\to1$ as $u\to\infty$, the permanent contribution of one isolated market order is $\kappa$, while $\xi(0)$ encodes the instantaneous amplification induced by order flow persistence. In the notation of Proposition~\ref{proposition:market_impact_functional}, the continuation functional is then explicitly
\begin{equation*}
    \mathcal R_t(q_t^a,q_t^b, N^a,N^b)
    =
    \kappa\int_0^t\bigl(\xi(t-s)-1\bigr)\,\dd\!\left(N^a_s-N^b_s\right).
\end{equation*}

\subsection{Functional counterfactual definition of market impact}
\label{market_impact_theory}

Market impact is the difference between the price generated by an intervened market and the price generated by a reference market in which the intervention is absent, with both markets compared under the same realized background randomness. To formalize this within Subsection~\ref{sec:modeling_prices}, suppose first that an oracle observes both the reference trajectory $\bigl((N_t^a,N_t^b,q_t^a,q_t^b)\bigr)_{t\ge 0}$ and an intervened trajectory $\bigl((\bar N_t^a,\bar N_t^b,\bar q_t^a,\bar q_t^b)\bigr)_{t\ge 0}$ constructed on the same filtered probability space $(\Omega, (\mathcal G_t)_{t \geq 0}, \mathbb{P})$. Define the price functional
\begin{equation*}
\mathcal P_t(q^a,q^b,N^a,N^b)
:=
P_0+
\int_0^t \kappa(q_s^a)\,\dd N_s^a
-\int_0^t \kappa(q_s^b)\,\dd N_s^b
+\mathcal R_t(q_t^a,q_t^b,S_t(N^a),S_t(N^b)),
\end{equation*}
so that, by Proposition~\ref{proposition:market_impact_functional}, $P_t=\mathcal P_t(q^a,q^b,N^a,N^b)$ a.s. The intervened price is then $\bar P_t:= \mathcal P_t(\bar q^a,\bar q^b,\bar N^a,\bar N^b)$ whenever the same pricing rule is evaluated on the intervened trajectory, and the corresponding pathwise market impact is
\begin{equation}
\label{eq:mi_functional_difference}
\mathrm{MI}_t
:=\bar P_t-P_t
=
\mathcal P_t(\bar q^a,\bar q^b,\bar N^a,\bar N^b)
-
\mathcal P_t(q^a,q^b,N^a,N^b).
\end{equation}


In the constant-$\kappa$ Hawkes model, we retrieve Equation~\eqref{eq:price_hawkes_propagator}.
If a buy metaorder $N^o$ is inserted as additional ask-side market orders and is assumed not to excite the background Hawkes flow, then $\bar N^a=N^a+N^o$, $\bar N^b=N^b$, and
\begin{equation*}
\mathrm{MI}_t=\kappa\int_0^t \xi(t-s)\,\dd N^o_s.
\end{equation*}
Thus each trade at time $s$ contributes $\kappa\xi(t-s)$ at time $t$, with permanent contribution $\kappa$ when $\xi(u)\to1$. \\

In practice no oracle exists: only one of the two trajectories is observed in the market. Making sense of Equation \eqref{eq:mi_functional_difference} therefore requires a model that links them, and the right notion of link is a coupling under shared randomness: factual and counterfactual processes must be generated on the same probability space, driven by the same exogenous market shocks, with the intervention modifying only the response of the queue and order flow to those shocks. The next subsection makes this coupling explicit and derives the resulting passive and aggressive impact formulas.
\subsection{From factual to counterfactual orders}

For aggressive orders, additional orders in the counterfactual market order flow are simply considered as an independent addition to the factual market order flow, as in \citep{jaisson2015market}. These new aggressive orders imply modifications of the queue dynamics. These dynamics are also influenced by counterfactual limit and cancel orders.
To model this feedback, order flow events $L^x$, $C^x$, and $N^x$ can be constructed from primitive Poisson point measures. More precisely, there exist independent Poisson point measures
\begin{equation*}
\Pi:=\bigl(\pi^{L,a},\pi^{C,a},\pi^{N,a},\pi^{L,b},\pi^{C,b},\pi^{N,b}\bigr)
\end{equation*}
on $\R_+\times\R_+$ with Lebesgue intensity such that, for $x\in\{a,b\}$,
\begin{align*}
L^x_t&=\int_0^t\!\!\int_0^\infty \1_{\{z\le \lambda^L(q^x_{u-})\}}\,\pi^{L,x}(\dd u,\dd z),
&
C^x_t&=\int_0^t\!\!\int_0^\infty \1_{\{z\le \lambda^C(q^x_{u-})\}}\,\pi^{C,x}(\dd u,\dd z),
\\
N^x_t&=\int_0^t\!\!\int_0^\infty \1_{\{z\le \lambda^x_u\}}\,\pi^{N,x}(\dd u,\dd z), & q^x_t &= L^x_t - C^x_t - N^x_t.
\end{align*}
where $\lambda^x_u=\mu+\int_0^{u-}\varphi(u-s)\,\dd N^x_s$.
This allows to define the coupling between factual and counterfactual Point processes based on the same underlying noise, formulated in Assumption \ref{assumption:shared_noise}.
\begin{assumption}[Shared-noise coupling]
\label{assumption:shared_noise}
Counterfactual order flow events, therefore denoted $\wb L^x$, $\wb C^x$, and thus the counterfactual queue process $\wb q^x$, are constructed from the same measure $\Pi$ as $L^x$, $C^x$, and $q^x$. Passive and aggressive interventions are then inserted as finite-variation paths in the corresponding queue equations, as specified below in Equations \eqref{eq:bar_queue}-\eqref{eq:bar_L_C}.
\end{assumption}

Under assumption \ref{assumption:shared_noise}, we can define the factual observed filtration $\mathcal{F}_t = \sigma(q^a_{\leq t}, q^b_{\leq t}, N^a_{\leq t}, N^b_{\leq t})$ and the shared more general filtration $\mathcal{G}_t = \sigma(\Pi_{|[0,t]})$. Assumption~\ref{assumption:shared_noise} is the precise analogue of the coupling used in \citep{youssef2024passiveimpact}. Economically, it expresses that
both worlds are exposed to the same latent market randomness: only the acceptance regions change across the two worlds; the underlying atoms are held fixed. Differences between factual and counterfactual prices are therefore attributable to execution rather than to an artificial resampling of exogenous noise. 
The aggressive formula below adopts the propagator convention used in \citep{jaisson2015market}: executed metaorder trades enter the ask side market order path appearing in the price functional up to the valuation time, while the ordinary background Hawkes path remains driven by the same $\pi^{N,a}$.
This convention is a modeling choice about the pricing rule, not a resampling of the primitive noise: the ordinary market order atoms are still those of $\pi^{N,a}$, but the signed order flow path supplied to the price functional includes the realized metaorder trades when the aggressive impact formula is evaluated.



\subsection{Oracle market impact formulas under common Poisson point measures}
\label{sec:market_impact_formula}

In this subsection, we derive the oracle formulas under the shared-noise coupling of
Assumption~\ref{assumption:shared_noise}. The passive case is special: because
the market order flow is assumed unchanged by passive interventions, the price
difference can be rewritten as the conditional expectation in
Equation~\eqref{eq:mi_passive_conditional}. This conditional-expectation form is
the basis for the Monte Carlo evaluation of passive impact. For aggressive orders,
by contrast, the intervention enters the market order path supplied to the pricing
functional; the impact therefore remains a difference of realized terms and
continuation values, as in Equation~\eqref{eq:mi_active_conditional}, and is not
reducible to a single conditional expectation without an additional approximation.
Formulas are stated on the ask side, bid side being symmetric. Throughout, Assumptions~\ref{assumption:passive} and~\ref{assumption:shared_noise} are in force.\\

\textbf{Oracle passive impact.}
Let $L^o$ be a cumulative ask-side limit metaorder. Since future intervention arrivals are not known at valuation time $t$, set
\begin{equation*}
L^{o,t}_s:=L^o_{s\wedge t},\qquad s\ge 0.
\end{equation*}
The path $L^{o,t}$ is treated as fixed at valuation time, or equivalently as conditioned upon together with the trader's realized strategy up to $t$.
The counterfactual ask queue is
\begin{equation}
\label{eq:bar_queue}
\bar q^{a,t}_s
=
q_0^a+\bar L^{a,t}_s-\bar C^{a,t}_s-N^a_s+L^{o,t}_s,
\qquad s\ge 0,
\end{equation}
where
\begin{equation}
\label{eq:bar_L_C}
\bar L^{a,t}_s=\int_0^s\!\!\int_0^\infty \1_{\{z\le \lambda^L(\bar q^{a,t}_{u-})\}}\,\pi^{L,a}(\dd u,\dd z),
\quad
\bar C^{a,t}_s=\int_0^s\!\!\int_0^\infty \1_{\{z\le \lambda^C(\bar q^{a,t}_{u-})\}}\,\pi^{C,a}(\dd u,\dd z).
\end{equation}
The bid side and market order flows are unchanged. Hence
\begin{equation*}
\bar P_t^{\mathrm{pas}}
= 
P_0+\lim_{T\to\infty}\E\!\left[
\int_0^T \kappa(\bar q^{a,t}_s)\,\dd N^a_s
-\int_0^T \kappa(q^b_s)\,\dd N^b_s
\ \middle|\ \calG_t\right],
\end{equation*}
and subtraction of Equation \eqref{eq:price} yields
\begin{equation}
\label{eq:mi_passive_conditional}
\mathrm{MI}_t^{\mathrm{pas}}
=
\E\!\left[
\int_0^\infty\!\big(\kappa(\bar q^{a,t}_s)-\kappa(q^a_s)\big)\,\dd N^a_s
\ \middle|\ \calG_t
\right].
\end{equation}
The monotone coupling of \citep{youssef2024passiveimpact} gives $\bar q^{a,t}_s\ge q^a_s$ for an ask side passive insertion; since $\kappa$ is non-increasing, the integrand is non-positive. Finiteness follows from the well-posedness of Equation \eqref{eq:price}.
Thus an ask side passive insertion has non-positive impact in this sign convention: it adds displayed ask liquidity and reduces the subsequent upward price response of buy market orders. The opposite side is obtained by symmetry.\\

\textbf{Oracle aggressive impact.}
Let $N^o$ be an ask-side market metaorder and set $N^{o,t}_s:=N^o_{s\wedge t}$. We keep the same ordinary market order noise $\pi^{N,a}$ and evaluate the signed-flow price functional on the total ask-side market order path $N^a+N^{o,t}$, consistently with the convention stated above.
The counterfactual queue is
\begin{equation}
\label{eq:bar_queue_aggressive}
\bar q^{a,t}_s
=
q_0^a+\bar L^{a,t}_s-\bar C^{a,t}_s-N^a_s-N^{o,t}_s,
\qquad s\ge 0,
\end{equation}
with $\bar L^{a,t}$ and $\bar C^{a,t}$ defined as in Equation \eqref{eq:bar_L_C}. Proposition~\ref{proposition:market_impact_functional} gives
\begin{equation*}
\bar P_t^{\mathrm{agg}}
=
P_0+
\int_0^t \kappa(\bar q^{a,t}_s)\,\dd \bigl(N^a_s+N^{o,t}_s\bigr)
-
\int_0^t \kappa(q^b_s)\,\dd N^b_s
+\mathcal R_t(\bar q^{a,t}_t,q^{b}_t,N^a+N^{o,t},N^b),
\end{equation*}
and therefore
\begin{equation}
\label{eq:mi_active_conditional}
\mathrm{MI}_t^{\mathrm{agg}}
=
\int_0^t \!\bigl(\kappa(\bar q^{a,t}_s)-\kappa(q^a_s)\bigr)\,\dd N^a_s
+\int_0^t \kappa(\bar q^{a,t}_s)\,\dd N^{o,t}_s
+\mathcal R_t(\bar q^{a,t}_t,q^{b}_t,N^a+N^{o,t},N^b)
-\mathcal R_t(q^a_t,q^b_t,N^a,N^b).
\end{equation}
The terms are, respectively, feedback on ordinary ask side trades, direct metaorder impact, and the change in continuation value. Unlike Equation \eqref{eq:mi_passive_conditional}, this is not reducible to a single conditional expectation against the unchanged ask flow without an additional approximation.

Relying on the functional formulation of market impact, one can moreover extend these equations by combining passive and aggressive orders in a given execution strategy.

\subsection{Latent-process formulation of impact evaluation}
\label{market_impact_evaluation}

Equations~\eqref{eq:mi_passive_conditional} and~\eqref{eq:mi_active_conditional}
are oracle formulas under the common Poisson source
$\Pi=(\pi^{L,a},\pi^{C,a},\pi^{N,a},\pi^{L,b}, \\ \pi^{C,b}, \pi^{N,b})$, with the
shared-noise coupling specified in Assumption~\ref{assumption:shared_noise}.
They are exact if $\Pi$ is known, but they are not directly computable from the
observed market history. Indeed, the counterfactual queue path $\bar q^{a,t}$
must be generated from the same latent Poisson atoms that produced the factual
paths, while the observed filtration $\calF_t$ contains only the aggregated paths
$(q^a,q^b,N^a,N^b)$ and not the marked measures $\Pi$ themselves. This is the
latent-noise reconstruction problem addressed in the rest of the paper.\\

We now formalize this point: factual and counterfactual prices are measurable
functionals of the same latent object $\Pi$, and market impact is the difference
of those functionals.
 The thinning equations of Assumption~\ref{assumption:shared_noise} (together with the Hawkes representation of $N^x$ used in the proof of Proposition~\ref{proposition:market_impact_functional}) yield measurable maps $\Phi^{N,x}_t,\Phi^{L,x}_t,\Phi^{C,x}_t$ such that, for $x\in\{a,b\}$ and every $t\ge 0$,
\begin{equation*}
N_t^x=\Phi_t^{N,x}(\Pi),\qquad L_t^x=\Phi_t^{L,x}(\Pi),\qquad C_t^x=\Phi_t^{C,x}(\Pi),
\end{equation*}
and hence, by Equation \eqref{eq:queue},
\begin{equation*}
q_t^x=\Phi_t^{q,x}(\Pi):=q_0^x+\Phi_t^{L,x}(\Pi)-\Phi_t^{C,x}(\Pi)-\Phi_t^{N,x}(\Pi).
\end{equation*}
The factual path $(q^a,q^b,N^a,N^b)$ is therefore a deterministic measurable functional of $\Pi$; composing with Proposition~\ref{proposition:market_impact_functional} gives a measurable map $\Psi_t$ with $P_t=\Psi_t(\Pi)$ a.s. Under the shared-noise counterfactual construction, the intervened processes and the counterfactual price are built from the same $\Pi$ (and the truncated intervention path), so there is likewise a measurable map $\bar\Psi_t$ with $\bar P_t=\bar\Psi_t(\Pi)$ a.s. Here the intervention path is regarded as fixed, or conditionally fixed once the trader's strategy is specified; if it is randomized, the same statement holds conditionally on that path. Market impact itself is thus a measurable functional of $\Pi$,
\begin{equation*}
\mathrm{MI}_t=\bar P_t-P_t=\bigl(\bar\Psi_t-\Psi_t\bigr)(\Pi),
\end{equation*}
consistently with Equation~\eqref{eq:mi_functional_difference} and, in the ask-side passive/aggressive cases, with Equations \eqref{eq:mi_passive_conditional}--\eqref{eq:mi_active_conditional}.\\

Thus $\Pi$ is the key latent object of the paper.
It is the common source of randomness from which factual and counterfactual worlds are simultaneously defined, but it is not directly observed: the factual filtration $\mathcal F_t$ contains only the aggregated counting paths $(q,N)$, which are strictly coarser than the marked measures $\Pi$. Consequently, the inverse map $(q,N)\mapsto \Pi$ is not pathwise available in general, and the observed data do not identify a single pathwise value of $\mathrm{MI}_t=(\bar\Psi_t-\Psi_t)(\Pi)$ without an additional reconstruction of the latent noise.\\

The remedy, and the technical core of the rest of the paper, is to replace the unknown $\Pi$ by its conditional law given the observation. The programme is the following: characterize the conditional distribution of $\Pi$ given the observed thinned event trajectory on the relevant horizon; sample from this conditional law; and push each draw through the measurable functional $\bar\Psi_t-\Psi_t$. The Monte Carlo distribution obtained this way is, by construction, the conditional law of $\mathrm{MI}_t$ given the realized market history, and yields in particular its conditional mean, variance, and quantiles. The next sections develop the two ingredients required by this programme: the conditional-law characterization (Section~\ref{sec:primer}, with proofs in Appendix~\ref{sec:conditional_measures}) and the resulting event-driven sampling procedure (Section~\ref{sec:quantifying}).

\section{Poisson-thinning framework and coupled trajectories}
\label{sec:primer}

This section fixes the probabilistic object held common in the factual and counterfactual markets of Section~\ref{sec:impact}. We use the Poisson random measure thinning representation of point processes with predictable intensities, see Appendix \ref{appendix:poisson:measures}, to formulate the shared-noise perturbation. Moreover, we state the conditional law of the latent Poisson measure given the observed thinned trajectory.

\subsection{Common randomness coupling for perturbed point process dynamics}

The primitive randomness in the coupling is the marked Poisson measure, not the thinned path. Once $\pi$ is fixed, a perturbation changes the acceptance boundary but not the candidate atoms. This is the common source used to compare the factual and counterfactual trajectories in Section~\ref{sec:impact}. We work in this section in a general setting for our statements, we use strong Poissonian representations. In dimension one, we consider a generic point process
\begin{equation*}
\label{eq:poisson_embedding_one_dimensional}
N_t=\int_0^t\int_{\R_+}\ind_{\{z\le\lambda_s\}}\pi(\dd s,\dd z).
\end{equation*}
The class needed below is
\begin{equation}
\label{eq:path_dependent_thinning}
N_t=\int_0^t\int_{\R_+}\ind_{\{z\le\lambda_s\}}\pi(\dd s,\dd z),
\qquad
\lambda_t=f(N,t),
\end{equation}
where $f:\mathcal D\times\R_+\to\R_+$ is measurable and predictable, or non-anticipative, in the path argument: $f(x,t)$ depends only on the pre-$t$ history of $x$. We assume that Equation \eqref{eq:path_dependent_thinning} admits a unique non-explosive strong solution for the driving $\pi$. For the Hawkes and queue-reactive specifications used later, this is supplied by the usual stability and local boundedness conditions \cite{hawkes1971spectra,bremaud1996stability,bacry2015hawkes,huang2015simulating}.

\begin{proposition}[Common-noise perturbation]
    \label{prop:shocked}
    Let $N^o$ be a deterministic non-decreasing càdlàg integer-valued process with $N^o_0=0$. If the perturbed equation is strongly well posed on the same probability space, the common-noise perturbed trajectory is the solution $(\widebar N,\widebar\lambda)$ of
    \begin{equation}
    \label{eq:common_noise_perturbed_count}
    \widebar N_t
    =
    \int_0^t\int_{\R_+}\ind_{\{z\le\widebar\lambda_s\}}\pi(\dd s,\dd z)
    +N^o_t,
    \qquad
    \widebar\lambda_t=f(\widebar N,t).
    \end{equation}
\end{proposition}

\begin{example}[Perturbed Hawkes process]
    Figure~\ref{fig:shocked_process} shows the process from Equation \eqref{eq:common_noise_perturbed_count} after adding the impulse $N^o_t=\1_{\{t\ge1\}}$ to a Hawkes trajectory. The same atoms are reused; only the intensity, hence the accepted hypograph, is recomputed from the perturbed history.
\end{example}

\begin{figure}[!htbp]
    \centering
    \includegraphics[width=0.9\textwidth]{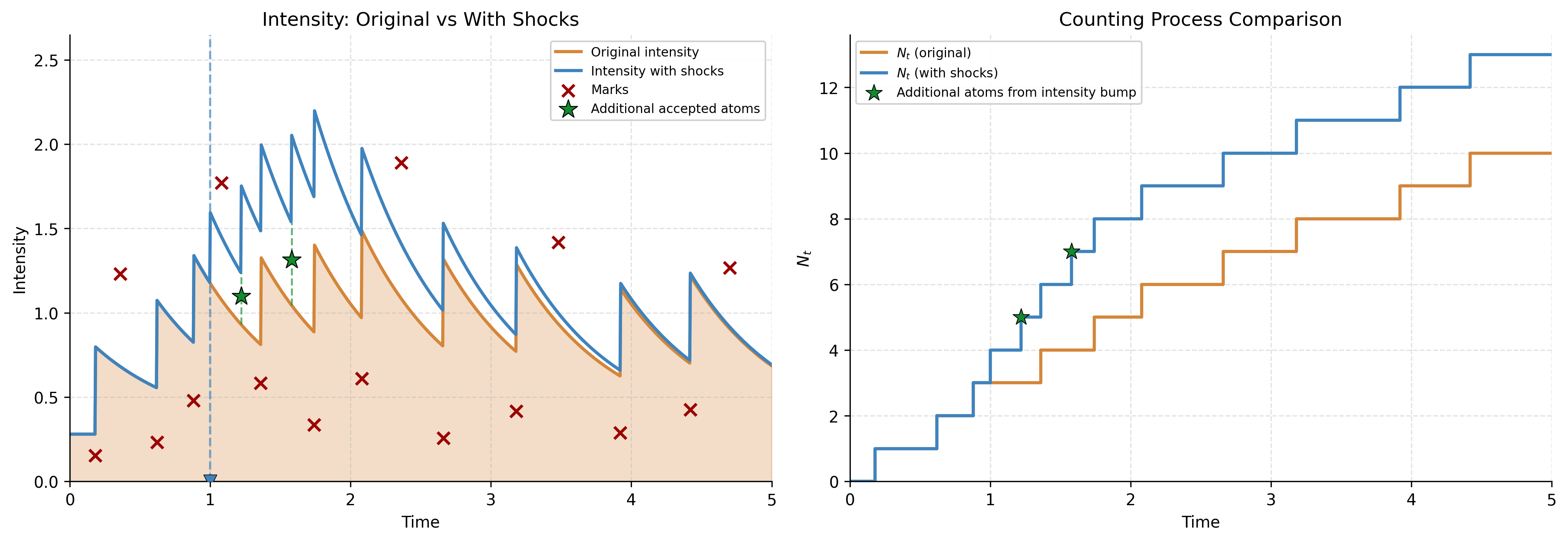}
    \caption{Impulse perturbation of a Hawkes process under shared Poisson noise.}
    \label{fig:shocked_process}
\end{figure}

Thus factual and perturbed paths are deterministic functionals of the same marked configuration. In applications this configuration is latent: the observed path reveals accepted jump times, but not the rejected atoms nor the accepted vertical marks. Counterfactual simulation therefore requires the conditional law of $\pi$ given the thinned trajectory.

\subsection{Conditional distribution of the underlying marked Poisson measure}
\label{sec:cond_distribution}

Fix $T>0$ and $d\in\mathbb N^*$. Set
\begin{equation*}
E:=[0,T]\times\{1,\ldots,d\}\times\R_+,
\qquad
\nu(\dd s,\dd k,\dd z):=\dd s\otimes\eta(\dd k)\otimes\dd z,
\end{equation*}
where $\eta(\{k\})=1$ for $1\le k\le d$. Let $e_k$ be the $k$-th canonical vector of $\R^d$.

For every counting measure $\mu\in\mathbb M(E)$ (see Appendix \ref{def:counting_measure} for the precise definition), assume a measurable strong solution
\begin{equation*}
\mu\longmapsto(\lambda^\mu,N^\mu)
\end{equation*}
with $N^\mu\in\calD([0,T],\R^d)$ and càglàd non-negative $\lambda^\mu$, satisfying
\begin{equation}
\label{eq:multivariate_thinning_functional}
N_t^\mu
=
\int_E e_k\,\ind_{\{s\le t\}}\ind_{\{z\le\lambda_s^{\mu,k}\}}\,\mu(\dd s,\dd k,\dd z),
\qquad
\lambda_t^\mu=f(N^\mu,t),
\end{equation}
for a measurable predictable, non-anticipative functional $f$. We also impose the finite-horizon bound
\begin{equation}
\label{eq:finite_horizon_non_explosion}
\forall\mu\in\mathbb M(E),\qquad
\sup_{t\in[0,T]}\sup_{1\le k\le d}\lambda_t^{\mu,k}<\infty.
\end{equation}
In the queue and Hawkes models of the paper, this abstract bound is replaced by the corresponding stability and boundedness assumptions from Sections~\ref{sec:impact} and~\ref{market_impact_evaluation}.\\

Let $\pi\sim\Poisson(\nu)$ on $E$, and write $(\lambda^\pi,N^\pi)$ for Equation \eqref{eq:multivariate_thinning_functional}. Define
\begin{equation*}
A(\mu):=\{(s,k,z)\in E:\ z\le\lambda_s^{\mu,k}\},
\qquad
B(\mu):=\overline{A(\mu)}.
\end{equation*}
The observed $\sigma$-field over $[0,t]$ for $t \leq T$ is
\begin{equation*}
\mathcal F_t^{N^\pi}:=\sigma((N_s^\pi)_{0\le s\le t}).
\end{equation*}
By Equation \eqref{eq:finite_horizon_non_explosion}, $A(\pi)$ has finite $\nu$-measure. Since $\lambda^\pi$ is càglàd, $\nu(B(\pi)\setminus A(\pi))=0$; hence $B(\pi)$ may replace $A(\pi)$ in stopping set arguments without adding Poisson atoms a.s.

\begin{restatable}{theorem}{conditionallaw}
    \label{thm:conditionallaw}
    The restrictions $\pi_{|B(\pi)}$ and $\pi_{|B(\pi)^c}$ are conditionally independent given $\mathcal F_T^{N^\pi}$. Moreover,
    \begin{equation}
    \label{eq:conditional_residual_poisson}
        \mathcal L\bigl(\pi_{|B(\pi)^c}\mid\mathcal F_T^{N^\pi}\bigr)
        =
        \Poisson\bigl(\ind_{B(\pi)^c}\nu\bigr).
    \end{equation}
    Finally, write the accepted atoms as
    \begin{equation}
    \label{eq:accepted_atoms_decomposition}
        \pi_{|A(\pi)}
        =
        \sum_{n\in I}\delta_{(T_n,k_n,Z_n)},
    \end{equation}
    where $I$ is finite a.s. Then $\{(T_n,k_n)\}_{n\in I}$ is $\mathcal F_T^{N^\pi}$-measurable, and conditionally on $\mathcal F_T^{N^\pi}$ we have
    \begin{equation}
    \label{eq:conditional_uniform_marks}
        \mathcal L\bigl((Z_n)_{n\in I}\mid\mathcal F_T^{N^\pi}\bigr)
        =
        \bigotimes_{n\in I}
        \mathcal U\bigl([0,\lambda_{T_n}^{\pi,k_n}]\bigr).
    \end{equation}
\end{restatable}


The simulation content of Theorem~\ref{thm:conditionallaw} is explicit. Given $N^\pi$, keep the observed jump times and components, resample each accepted mark independently and uniformly on its admissible interval $[0,\lambda_{T_n}^{\pi,k_n}]$, and add an independent Poisson cloud on $B(\pi)^c$ with intensity $\ind_{B(\pi)^c}\nu$. Solving the perturbed dynamics on this reconstructed configuration gives a counterfactual path conditionally compatible with the factual one.\\

For deterministic Borel set $C$, Poisson scattering gives independent restrictions on $C$ and $C^c$. Theorem~\ref{thm:conditionallaw} is subtler because $B(\pi)$ is random and depends on $\pi$ through the solution of the thinning equation. The required additional property is locality: changing atoms outside $B(\pi)$ does not change the accepted path, hence does not change $B(\pi)$. Appendix~\ref{appendix:section_three_tree} proves this as a stopping set property and then applies the strong Markov theorem for Poisson point processes \cite[Theorem~12.1.3]{RPG}.

\begin{example}[Failure of conditional independence without locality]
    Let $\pi$ be a Poisson random measure on $[0,1]^2$ with Lebesgue intensity, let $M:=\pi([0,1]^2)$, and, on $\{M\ge2\}$, order the atoms by decreasing ordinate $Z_{(1)}>Z_{(2)}>\cdots>Z_{(M)}$, with corresponding abscissas $X_{(i)}$. Define
    \begin{equation*}
    A(\pi):=
    \begin{cases}
        \{(x,y)\in[0,1]^2:\ x=X_{(2)}\}, & M\ge2,\\
        \emptyset, & M<2.
    \end{cases}
    \end{equation*}
    This set is selected by a global ranking. For instance, $\{A(\pi)\subset[0,a]\times[0,1]\}=\{X_{(2)}\le a\}$ on $\{M\ge2\}$, which cannot be decided from the atoms inside $[0,a]\times[0,1]$ alone.\\

    \begin{figure}[!htbp]
        \centering
        \includegraphics[width=0.85\textwidth]{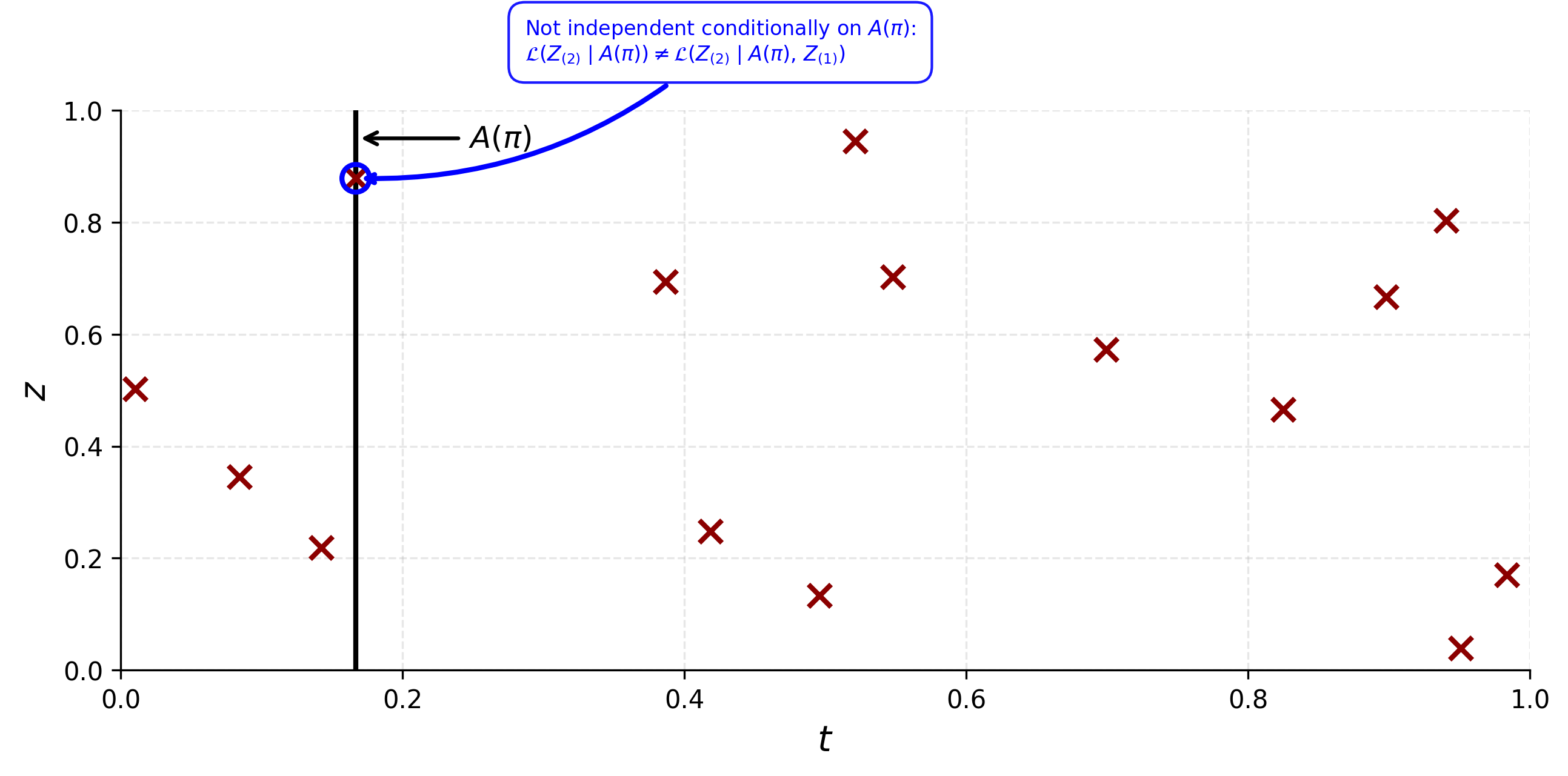}
        \caption{Visualizing the breakdown of independence due to global ordering.}
        \label{fig:independence_breakdown}
    \end{figure}

    Conditional independence fails: given $M=m\ge2$, the ordinate in $A(\pi)$ is the second largest of $m$ uniforms. If the complement reveals the largest ordinate $Z_{(1)}$, the same ordinate becomes the largest of $m-1$ uniforms on $[0,Z_{(1)}]$. Thus the conditional law inside $A(\pi)$ still depends on information carried by $A(\pi)^c$.
\end{example}


The conditional law Theorem~\ref{thm:conditionallaw} provides the density of the conditional Poisson measure and can be formulated in Laplace form as follows.
\begin{corollary}
    For every positive measurable $h:E\to\R_+$,
    \begin{equation*}
        \E \Big[ \exp \Big( -\int h \, \dd \pi \Big) \mid \calF^{N^{\pi}}_T \Big]
        =
        \exp \Big(
            -\int_{B(\pi)^c} (1 - e^{-h}) \, \dd \nu
        \Big)
        \prod_{n \in I} \frac{1}{\lambda_{T_n}^{\pi, k_n}} \int_0^{\lambda_{T_n}^{\pi, k_n}} e^{-h(T_n, k_n, z)} \, \dd z,
    \end{equation*}
    with the empty product equal to one.
\end{corollary}

\section{Estimation of market impact}
\label{sec:quantifying}

This section turns the conditional law results of Section~\ref{sec:cond_distribution} into an operational market impact workflow. We work on a fixed side of the book, with one sided notation (representing either the bid or ask events) $(L,C,N,q)$ for the factual queue and $(\wb L,\wb C,\wb N,\wb q)$ for the intervened queue. The objective is to reconstruct counterfactual paths under the same latent Poisson noise and then evaluate the impact functionals introduced in Section~\ref{sec:impact}. In other words, we seek to simulate the counterfactual, intervened dynamics $(\wb L,\wb C,\wb N,\wb q)$ conditionally on the observed no-trade realization $(L, C, N, q)$. From this, one can evaluate the market impact of candidate trading strategies, each replayed against the same observed realization.\\

We proceed in two steps. First, we present the conditional reconstruction mechanism itself for the bid and ask side independently, we focus on the ask side: this is the event-driven part of the method and only uses the thinning structure of the order flow components. We then combine this reconstruction with the price model of Section~\ref{sec:impact}, where market orders are Hawkes and the difference between the limit and cancel intensity functions is affine in $q$, in order to compute counterfactual price trajectories, and thus estimate passive and aggressive impacts. The figures in Subsection~\ref{subsec:conditional_queue_coupling} illustrate how the counterfactual queue is generated, while Subsections~\ref{subsec:first_order_impact} and~\ref{subsec:aggressive_market_impact} explain how these conditional paths are converted into impact estimates  .

\subsection{Conditional reconstruction of counterfactual queue paths}
\label{subsec:conditional_sim}
\label{sec:lob_sim}

Fix a horizon $T>0$ and consider the one-sided queue model,
\begin{equation*}
q_t=q_0+L_t-C_t-N_t,
\qquad
\wb q_t=q_0+\wb L_t-\wb C_t-\wb N_t+L_t^o,
\end{equation*}
where $L^o$ is the passive intervention (for aggressive intervention, replace $+L_t^o$ by $-N_t^o$). We condition on
\begin{equation*}
\mathcal F_T=\sigma\bigl((L_t,C_t,N_t,q_t),\,0\le t\le T\bigr),
\end{equation*}
and seek exact samples from $\mathcal L((\wb L,\wb C,\wb N,\wb q)\mid \mathcal F_T)$.

\begin{algorithm}[!htbp]
\caption{Conditional simulation of $(\wb L,\wb C,\wb N,\wb q)$ given $\mathcal F_T$}
\label{algo:conditional_simulation}
\KwIn{Observed trajectory $(L,C,N,q)$ on $[0,T]$; jump times $(T_j^x)_{j\ge1}$ for $x\in\{L,C,N\}$; intensities $\lambda^L,\lambda^C,\lambda^N$; intervention path $L^o$ (set $L^o\equiv 0$ if there is no passive intervention).}
\KwOut{One conditional draw of $(\wb L,\wb C,\wb N,\wb q)$ on $[0,T]$.}
Set $t\leftarrow 0$, $\wb L_0\leftarrow 0$, $\wb C_0\leftarrow 0$, $\wb N_0\leftarrow 0$, and $\wb q_0\leftarrow q_0$\;
For each $x\in\{L,C,N\}$ and each jump time $T_j^x$, draw and store $U_j^x\sim\mathcal U([0,1])$\;
\While{$t<T$}{For each $x\in\{L,C,N\}$, draw $\tau^x\sim\mathrm{Exp}((\lambda^x(\wb q_t)-\lambda^x(q_t))_+)$, with $\tau^x=+\infty$ if the rate is zero\;
Set $\tau^{L^o}$ to the time from $t$ to the next jump of $L^o$ (or $+\infty$ if no jump remains before $T$)\;
For each $x\in\{L,C,N\}$, let $\theta^x\leftarrow T_{j_x}^x-t$ be the waiting time to the next unprocessed observed jump of type $x$ (or $+\infty$ if none remains)\;
Set $\theta\leftarrow\min\{\theta^L,\theta^C,\theta^N\}$ and let $x^\star$ be the corresponding type\;
Set $\tau\leftarrow \min\{\tau^L,\tau^C,\tau^N,\tau^{L^o},\theta,T-t\}$ and $t\leftarrow t+\tau$\;
\If{$t=T$}{stop the loop\;}
\If{$\tau=\theta$}{
Let $j=j_{x^\star}$\;
\If{$U_j^{x^\star}\lambda^{x^\star}(q_{t-})\le \lambda^{x^\star}(\wb q_{t-})$}{
Increase the corresponding counter $\wb L$, $\wb C$, or $\wb N$ by one at time $t$\;
}
$j_{x^\star}\leftarrow j_{x^\star}+1$\;
}
\ElseIf{$\tau=\tau^{L^o}$}{
Apply the intervention jump through the known path $L^o$\;
}
\Else{
Let $x$ be such that $\tau=\tau^x$ and increase the corresponding counter $\wb L$, $\wb C$, or $\wb N$ by one at time $t$\;
}
Update $\wb q_t$ from $\wb q_t=q_0+\wb L_t-\wb C_t-\wb N_t+L_t^o$\;
}
\Return $(\wb L,\wb C,\wb N,\wb q)$ on $[0,T]$\;
\end{algorithm}

\begin{proposition}
\label{prop:conditional_simulation}

Assume the one-sided queue dynamics are driven by thinning on a common Poisson noise as in Sections~\ref{sec:primer} and~\ref{sec:conditional_measures}. Then Algorithm~\ref{algo:conditional_simulation} generates an exact sample from
\begin{equation*}
\mathcal L\big((\wb L,\wb C,\wb N,\wb q)\mid\mathcal F_T\big).
\end{equation*}
\end{proposition}

The proof is given in Appendix~\ref{appendix:B}. The decomposition of candidate times has a direct interpretation: the clocks $(\tau^L,\tau^C,\tau^N)$ create extra atoms induced by intensity gaps, $(\wb\tau^L,\wb\tau^C,\wb\tau^N)$ recycle observed atoms through conditional thinning, and $\tau^{L^o}$ carries the intervention flow.

\begin{remark}
    For an aggressive intervention, the same construction applies after replacing the intervention clock $\tau^{L^o}$ by the market order clock $\tau^{N^o}$ and updating the queue with $-N^o$. Figure \ref{fig:queue_with_market} illustrates the resulting simulated queue dynamics.
\end{remark}

\begin{remark}
    In practice, one fixes an observed path, calibrates $\lambda^L,\lambda^C,\lambda^N$, and repeats Proposition~\ref{prop:conditional_simulation} to obtain conditional replicas. The procedure is event-driven and linear in the number of simulated jumps, so Monte Carlo replications are straightforward to parallelize. This is the core numerical ingredient used below for impact quantification; the corresponding implementation used for the numerical experiments is available at \githubrepo.
\end{remark}

\subsection{Conditional simulation of the coupled queues}
\label{subsec:conditional_queue_coupling}

We now illustrate the conditional coupling on an interval $[0,T]$, following Algorithm~\ref{algo:conditional_simulation}. Throughout the section, the simulation parameters are:
\begin{itemize}
    \item $\lambda^L(x) = 100-0.275x$,
    \item $\lambda^C(x) = 2+0.125x$,
    \item Market orders are jump times of Hawkes process with $\mu = 1.0$ and $\varphi(x) = 0.065 e^{-0.15x} + 0.2 e^{-0.6x} + 0.325 e^{-2.5x} + 0.65 e^{-10x}$, which corresponds to a near critical regime $\norm{\varphi}_1 = 0.96$,
    \item The simulation horizon $T$ is chosen to be 1.5 minutes, while the metaorder (limit or market) is executed over the first minute of the simulation window.
\end{itemize}
Figure~\ref{fig:queue_with_limit} shows the factual queue (black) together with the counterfactual simulated trajectories (grey) and their average (red), under a passive intervention in which the trader adds limit order volume to the ask queue. Here the trader executes a metaorder uniformly over one minute, of a size accounting for 10\% of the average limit order flow. Figure~\ref{fig:queue_with_market} shows the analogous results under an aggressive intervention, in which the trader consumes ask-side liquidity through market orders. Here too the metaorder is executed uniformly over one minute, of a size accounting for 10\% of the average market order flow. In both cases the observed path $q$ is held fixed, the atoms compatible with $q$ are reused through the conditional marks of Theorem~\ref{thm:conditionallaw}, and only the unrevealed residual Poisson measure is resampled.

\begin{figure}[htbp]
        \centering
        \includegraphics[width=0.85\textwidth]{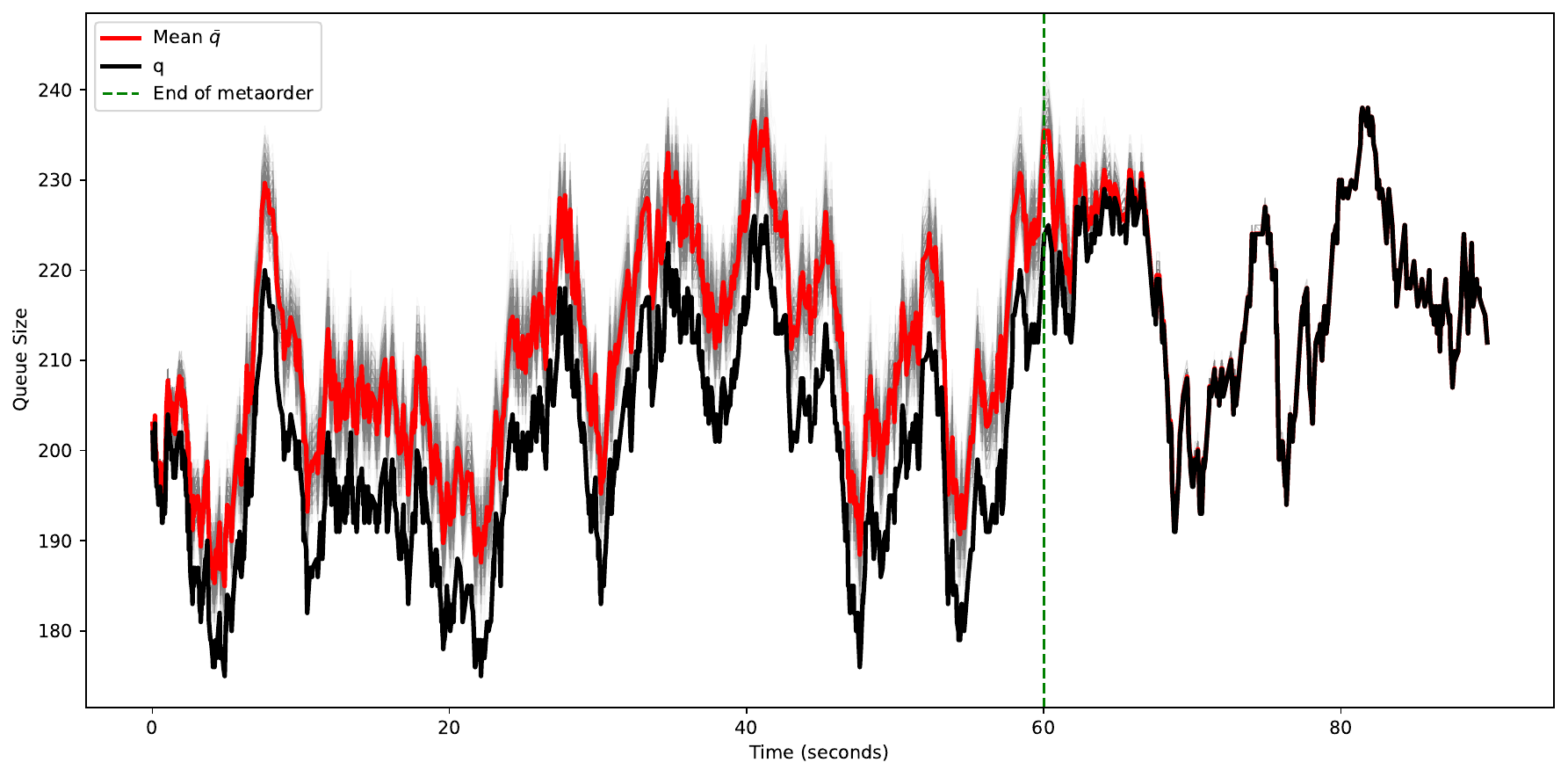}
        \caption{Conditional simulation of $\wb q$ given the observed baseline $q$ in the presence of a limit metaorder.}
        \label{fig:queue_with_limit}
\end{figure}

During the execution window, the passive metaorder in Figure \ref{fig:queue_with_limit} adds displayed liquidity, so the simulated intervened queues and their mean lie above the factual queue. The spread of the grey paths measures the remaining uncertainty on the latent Poisson atoms outside the revealed region. After the end of the metaorder, the drift pulls the perturbed and factual systems back toward the same regime. Once the intervention ceases, the two systems are governed by identical dynamics and, through the conditional coupling, are driven by the same residual randomness; all that distinguishes them is the state inherited at the moment execution stops. Because these dynamics are mean-reverting, that difference in initial condition is gradually forgotten: consumed liquidity is replenished, the imbalance created by the trade dissipates, and the perturbed trajectory relaxes back onto the factual one. The gap between the two paths, which is exactly the counterfactual impact, therefore decays at a rate set by the mean-reversion speed of the system, rather than persisting indefinitely. This relaxation is the microstructural signature of market resilience: the order book absorbs the metaorder and, left to its own dynamics, returns to equilibrium.\\

For the aggressive case, the counterfactual queue dynamics are given by
\begin{equation*}
\wb q_t = q_0 + \wb L_t - \wb C_t - \wb N_t - N^o_t,
\end{equation*}
with $N^o$ the aggressive metaorder flow. The same conditional simulation machinery applies after replacing the intervention clock, exactly as in Proposition~\ref{prop:conditional_simulation}.\\

\begin{figure}[!htbp]
        \centering
        \includegraphics[width=0.85\textwidth]{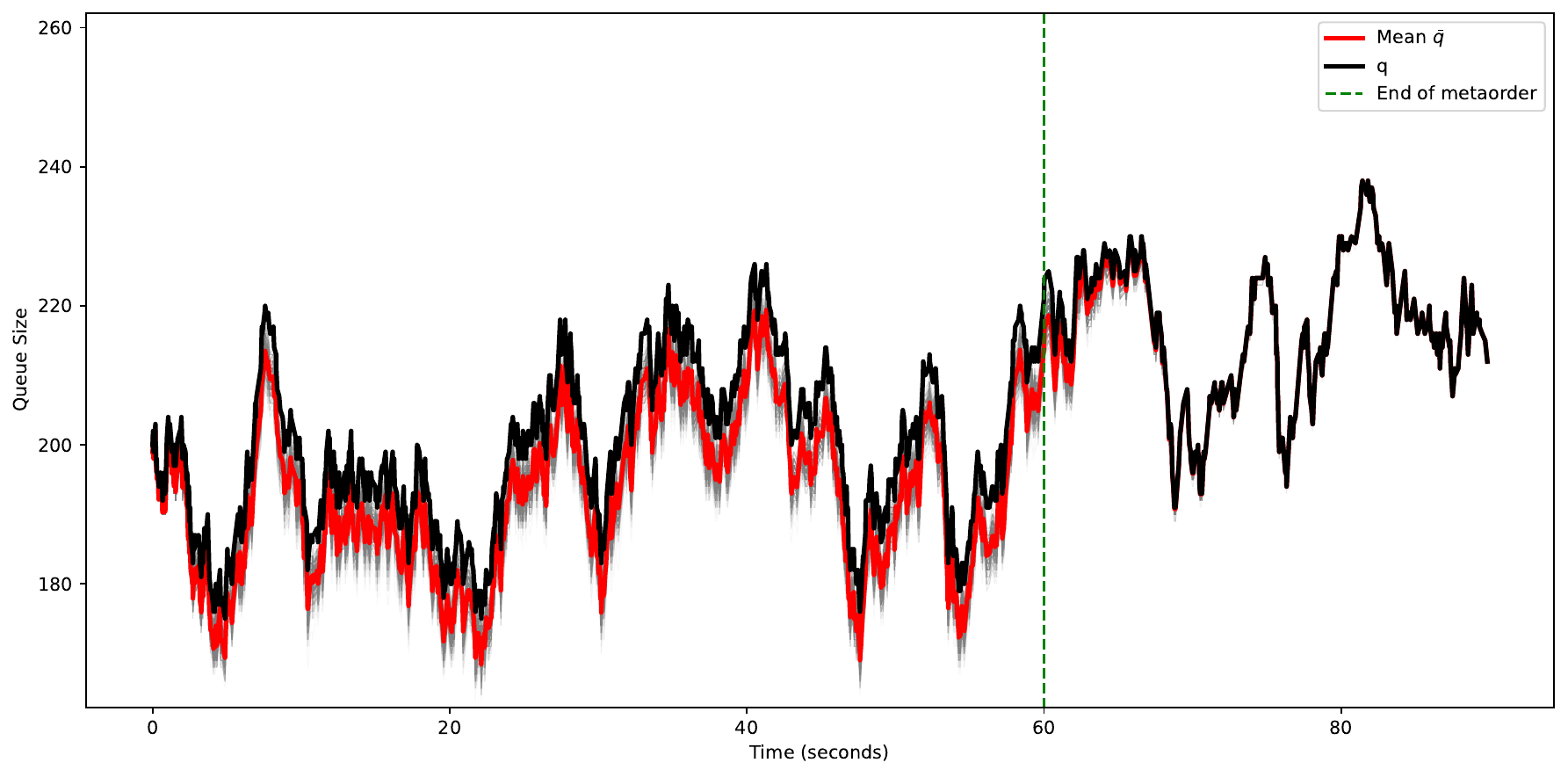}
        \caption{Conditional simulation of $\wb q$ given the observed baseline $q$ in the presence of a market metaorder.}
        \label{fig:queue_with_market}
    \end{figure}

In Figure~\ref{fig:queue_with_market}, the sign is reversed: the aggressive metaorder removes queue volume, so the conditional mean of the perturbed queue lies below the factual trajectory during most of the execution window. The narrowing of the trajectories after the vertical line has the same interpretation as in the passive case: once the intervention stops, the remaining discrepancy is carried forward only by the endogenous queue dynamics and the common future noise.

\begin{remark}
The fast reconvergence after the trading window is consistent with the coupling logic of Section~\ref{sec:impact}. First, monotonicity results from \citep{youssef2024passiveimpact} imply pathwise ordering during intervention ($\wb q\ge q$ for passive buy pressure, $\wb q\le q$ for aggressive buy pressure). Second, after the metaorder ends, both queues evolve with the same baseline dynamics and the same realized noise. Under the affine drift condition $\lambda^L(x)-\lambda^C(x)=c_\lambda x+d_\lambda$, with $c_\lambda<0$, the corresponding mean-field calculation, with $\tau_o$ the end of execution window, gives
\begin{equation*}
\E\!\left[|\Delta_t|\mid \calG_{\tau_o}\right]\lesssim e^{c_\lambda (t-\tau_o)}|\Delta_{\tau_o}|,
\qquad
\Delta_t:=\wb q_t-q_t,
\qquad t\ge \tau_o,
\end{equation*}
This is precisely why conditional coupling separates transient execution effects from background fluctuations.
\end{remark}

\subsection{Conditional estimation of market impact}
\label{subsec:first_order_impact}
Having described the reconstruction of counterfactual queue dynamics under a given trading strategy, executed through limit or market orders, we now turn to estimating the impact of such a strategy. Here we answer the question:
        \begin{center}
        \textit{Given a market realization observed in the past, what would have happened under a given strategy, and at what cost?}
        \end{center}
        \vspace{4mm}
For passive impact, at each time of interest, one has to run Monte Carlo estimates to compute the conditional expectation in Equation \eqref{eq:mi_passive_conditional} from each of the current counterfactual queue states. We will see however that under specific assumptions that it is possible to reduce computations to a simple Monte Carlo without nesting.

\subsubsection{Passive market impact}
To connect queue dynamics to market impact, we rely on the passive impact formula of Section~\ref{sec:market_impact_formula}. Under a limit-only metaorder executed on the ask side, at each valuation time $t$, define the non-anticipative truncation
\begin{equation*}
L_s^{o,t}:=L^o_{s\wedge t},\qquad s\ge 0,
\end{equation*}
and the corresponding ask queue
\begin{equation*}
\wb q_s^{a,t}=q_0^a+\wb L_s^{a,t}-\wb C_s^{a,t}-N_s^a+L_s^{o,t}.
\end{equation*}
This is motivated by the following observation: although the market digests the orders of the metaorder executed up to time $t$, it cannot anticipate those yet to arrive. Therefore, when computing the price at time $t$ as a conditional expectation of a functional of the order flow, only the portion of the metaorder executed prior to $t$ should be taken into account. The associated counterfactual price is
\begin{equation}
\label{eq:price:lim:metaorder}
\wb P_t = P_0 + \lim_{T\to\infty}
\E\bigg[\int_0^T \kappa(\wb q^{a,t}_{s}) \, \dd N^a_s - \int_0^T \kappa(q^b_s) \, \dd N^b_s \bigg | \calG_t \bigg].
\end{equation}
This is the passive counterpart of the general formulas derived in Section~\ref{sec:market_impact_formula}. The passive impact process is therefore
\begin{equation}
\label{eq:mi}
\mathrm{MI}_t^{\mathrm{pas}} = \wb P_t - P_t =
\E\bigg[\int_0^\infty \big( \kappa(\wb q^{a,t}_{s}) - \kappa(q^a_s) \big) \, \dd N^a_s \bigg | \calG_t \bigg].
\end{equation}
Equation~\eqref{eq:mi} is the quantity estimated numerically via conditional simulation. Henceforth, $N^a$ denotes the Hawkes market order flow of Section~\ref{sec:impact}; under a passive intervention, we assume that its intensity is independent of queue state and is not perturbed by $L^o$. In this framework, the impact of a strategy on the ask side therefore depends on three objects: the baseline queue trajectory $q^a$, the counterfactual trajectory $\wb q^a$ coupled to the same underlying noise, and the market order flow, which is assumed independent of the rest of the system. Estimating this impact thus reduces to simulating $\widebar{q}^a$ conditionally on the observed path $q^a$, precisely what the conditional simulation algorithm of Section~\ref{subsec:conditional_sim} provides, after which the impact of different trading strategies can be quantified a posteriori.\\

To obtain a closed-form that can be efficiently evaluated, we now impose the following affine conditions used in \citep{youssef2024passiveimpact} as well as the multi-exponential specification for the Hawkes process.

\begin{assumption}
\label{assum:linear}
The impact function and queue drift are affine and given by
\begin{equation*}
\kappa(x)=d_\kappa+c_\kappa x,
\qquad
\lambda^L(x)-\lambda^C(x)=c_\lambda x+d_\lambda,
\end{equation*}
with $c_\kappa<0$ and $c_\lambda<0$. Moreover, $N^a$ is Hawkes process with baseline intensity $\mu$ and self-exciting kernel
\begin{equation}
\label{eq:kernel_sum_exps}
\varphi(t)=\sum_{i=1}^m\alpha_i e^{-\beta_i t},
\end{equation}
for some $m\ge 1$, $\alpha_i>0$, $\beta_i>0$, and $\sum_{i=1}^m \alpha_i/\beta_i<1$.
\end{assumption}

Under Assumption~\ref{assum:linear}, the impact given by Equation \eqref{eq:mi} admits the following explicit representation.

\begin{theorem}
    \label{thm:explicit_impact_shape}

Under Assumption~\ref{assum:linear}, there exist constants $(\gamma_i)_{1\le i\le m}$ and $\zeta$ such that
\begin{equation*}
\mathrm{MI}_t^{\mathrm{pas}} = c_{\kappa} \int_0^t (\widebar{q}^{a,t}_s - q^{a}_s)\,\dd N^a_s
+ c_{\kappa} (\widebar{q}^{a,t}_t - q^{a}_t)\Big(\zeta + \int_0^t \sum_{i=1}^m \gamma_i e^{-\beta_i (t-s)}\,\dd N^a_s\Big).
\end{equation*}
\end{theorem}

\begin{remark}
Equivalently, define the effective passive-response kernel
\begin{equation*}
\widetilde{\xi}(u):=\sum_{i=1}^m \gamma_i e^{-\beta_i u},
\qquad u\ge0.
\end{equation*}
Thus Theorem~\ref{thm:explicit_impact_shape} can be written as
\begin{equation}
\label{eq:closed_form_mi_passive}
\mathrm{MI}_t^{\mathrm{pas}}
=
c_{\kappa} \int_0^t (\widebar{q}^{a,t}_s-q^a_s)\,\dd N^a_s
+c_{\kappa}(\widebar{q}^{a,t}_t-q^a_t)
\left(
\zeta+\int_0^t\widetilde{\xi}(t-s)\,\dd N^a_s
\right).
\end{equation}
The first term is the realized contribution of ask-side market orders up to time $t$. The second term is the continuation value of the queue displacement at time $t$.

Explicit computations give
\begin{equation*}
\gamma_i=\frac{\alpha_i}{D(\beta_i-c_\lambda)},
\qquad
\zeta=-\frac{\mu}{D c_\lambda},
\qquad
D=1-\sum_{i=1}^m\frac{\alpha_i}{\beta_i-c_\lambda}.
\end{equation*}

The lower limit $0$ in \eqref{eq:closed_form_mi_passive} is exact under the convention that the Hawkes process is initialized at time $0$, as in Assumption~\ref{assum:linear}. If instead time $0$ denotes the beginning of the metaorder or of the recorded sample, while the market order flow was already active before that time, the Hawkes prehistory should be retained. In the stationary version of Theorem~\ref{thm:explicit_impact_shape}, the continuation term is therefore evaluated with
\begin{equation*}
    c_{\kappa}\bigl(\widebar{q}^{a,t}_t-q^a_t\bigr)
    \left(
    \zeta+\int_{-\infty}^{t}\widetilde{\xi}(t-s)\,\dd N^a_s
    \right).
\end{equation*}
In this stationary formulation, one can also compute explicitly the unconditional mean impact 
\begin{equation*}
\E\!\left[\mathrm{MI}_t^{\mathrm{pas}}\right]
=
-\frac{\mu}{1-\|\varphi\|_{L^1}}\,
\frac{c_\kappa}{c_\lambda}\,
V_t,
\end{equation*}
where $V_t$ denotes the total passive volume posted up to time $t$. This identity shows that the expected passive impact is proportional to the cumulative posted volume, with proportionality constant determined by the stationary market order intensity, the price sensitivity $c_\kappa$, and the queue mean-reversion coefficient $c_\lambda$. Note that this quantity is consistent with \cite{youssef2024passiveimpact}.
\end{remark}

\begin{remark}
Theorem~\ref{thm:explicit_impact_shape} removes nested conditional expectations from online computation and turns impact simulation into a Markovian-factor update. If one starts from power-law Hawkes kernels, exponential-sum approximations (e.g. Beylkin--Monz\'on) can be used to remain in this tractable class; see \citep{bacry2014estimationslowlydecreasinghawkes,beylkin2010approximation}.
\end{remark}

For the simulation of a passive impact distribution, we:
\begin{itemize}
    \item generate one baseline ask queue trajectory $q^a$ (interpreted as the observed path);
    \item simulate a fixed realization of the passive metaorder $L^o$ for 1 minute, accounting for 10\% of the limit order flow at the ask, in a similar way to Section \ref{subsec:conditional_queue_coupling};
    \item condition on $(q^a,L^o)$ and simulate i.i.d. baseline trajectories $(\wb q^{a,i})_{1\le i\le n}$ with Proposition~\ref{prop:conditional_simulation};
    \item evaluate Equation \eqref{eq:closed_form_mi_passive} (or \eqref{eq:mi} if we ignore the assumptions of Theorem \ref{thm:explicit_impact_shape}) on each replica to approximate the conditional distribution of impact. The parameter $c_\kappa < 0$ is left implicit.
\end{itemize}
Figure~\ref{fig:efficient_impact_forward} shows the distribution of market impact estimates obtained via conditional simulation, as described above. After execution ends, the impact does not vanish: the posted limit orders create a lasting resistance to incoming trades, leaving a persistent price displacement consistent with the findings of \citep{youssef2024passiveimpact}. It is worth separating the two levels at which this convergence acts. The queue states themselves reconverge, so the perturbation is transient in the state variables; but the price is an integral of the queue (through $\kappa(q)$) along the path, and the displacement accumulated while the two systems differ need not vanish when the gap closes. Any permanent component of impact therefore lives in this accumulated integral, not in the queue dynamics, which forget the metaorder entirely. In short, the system forgets the trade at the level of its state, while the price may retain a lasting memory of it. From a more practical perspective, since limit orders and cancellations account for the overwhelming majority of order flow, a passive execution leaves only a faint signature in the observable activity, making its detection substantially harder than that of an aggressive strategy of comparable size.

     \begin{figure}[!htbp]
        \centering
        \includegraphics[width=0.85\textwidth]{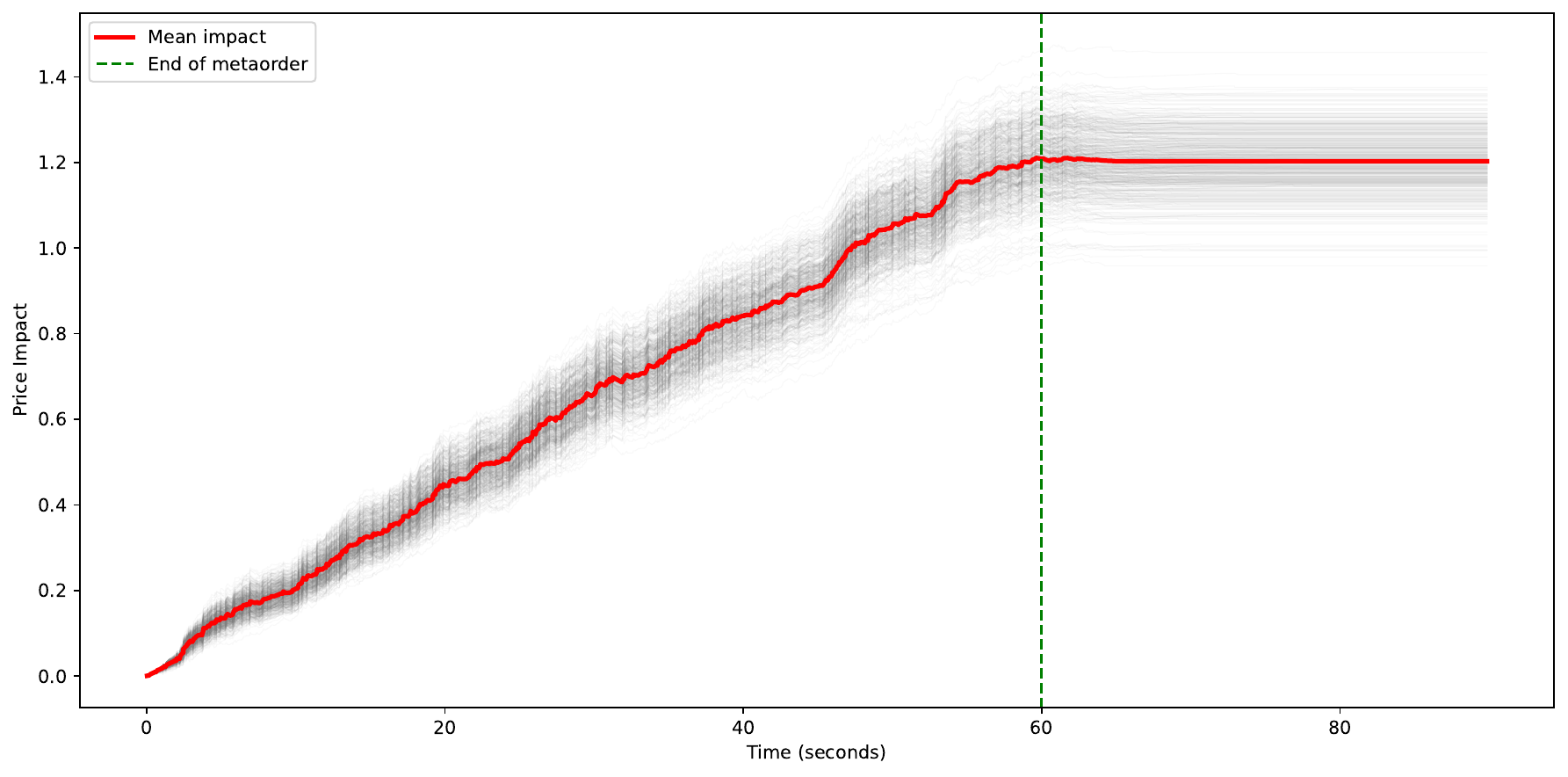}
        \caption{Conditional simulation of market impact given the observed baseline $q$ in the presence of a limit metaorder.}
        \label{fig:efficient_impact_forward}
    \end{figure}

\subsubsection{Aggressive market impact}
\label{subsec:aggressive_market_impact}
For aggressive execution, the exact formula of Equation \eqref{eq:mi_active_conditional} still contains the generally intractable continuation difference  $\mathcal R_t(\bar q^{a,t}_t,q^b_t,N^a+N^{o,t},N^b)-\mathcal R_t(q^a_t,q^b_t,N^a,N^b)$. We therefore use the reduced-form approximation introduced in \citep{hafsi2026impact}, replacing this difference by its propagator part plus observable queue corrections. The price approximation is then written as
\begin{equation*}
    P_t = P_0 + \bar{\kappa} \int_0^t \xi(t-s) d(N^a - N^b)_s + \int_0^t\bigl(\kappa(q^a_s)-\widebar{\kappa}\bigr)\,\dd N^a_s
-\int_0^t\bigl(\kappa(q^b_s)-\widebar{\kappa}\bigr)\,\dd N^b_s,
\end{equation*}
by neglecting the theoretical error term $J^\kappa_t := \int_t^\infty
\E\!\left[
\bigl(\kappa(q^a_s)-\widebar{\kappa}\bigr)\lambda^a_s
-\bigl(\kappa(q^b_s)-\widebar{\kappa}\bigr)\lambda^b_s
\ \middle|\ \calF_t
\right]\dd s$.

Given a buy market metaorder $N^o$ on the ask side, we define
\begin{equation*}
\label{eq:impact_aggressive}
    \mathrm{MI}_t^{\mathrm{agg}}
    =
    \int_0^t \left(\kappa(\wb q^{a}_s) - \kappa(q^a_s)\right)\,\dd N^a_s
    +
    \int_0^t
    \left(
    \widebar{\kappa}\,\xi(t-s)+\kappa(\wb q^{a}_s)-\widebar{\kappa}
    \right)\,\dd N^o_s
\end{equation*}
with counterfactual queue dynamics
\begin{equation*}
\wb q^a_t = q_0^a + \wb L^{a}_t - \wb C^{a}_t - N^{a}_t - N^{o}_t.
\end{equation*}

The first term in Equation \eqref{eq:impact_aggressive} is the indirect queue-feedback effect on ordinary ask-side market orders. The second term is the direct contribution of the metaorder itself: the factor $\widebar{\kappa}\xi(t-s)$ is the propagator response of an additional buy market order, and $\kappa(\wb q^a_s)-\widebar{\kappa}$ corrects this response for the contemporaneous ask-queue state.\\

The simulation loop is the same as in the passive case: simulate one baseline path of $q$ and one realization of the metaorder $N^o$, condition on it, generate conditional baseline replicas under shared noise, then compute the impact functional pathwise. The simulation is again done over 2 minutes, with the metaorder executed over the first minute and accounting for 10\% of the typical market order flow at the ask. Figure~\ref{fig:mi_agressive} reports the resulting distribution.
\begin{figure}[!htbp]
    \centering
    \includegraphics[width=0.85\textwidth]{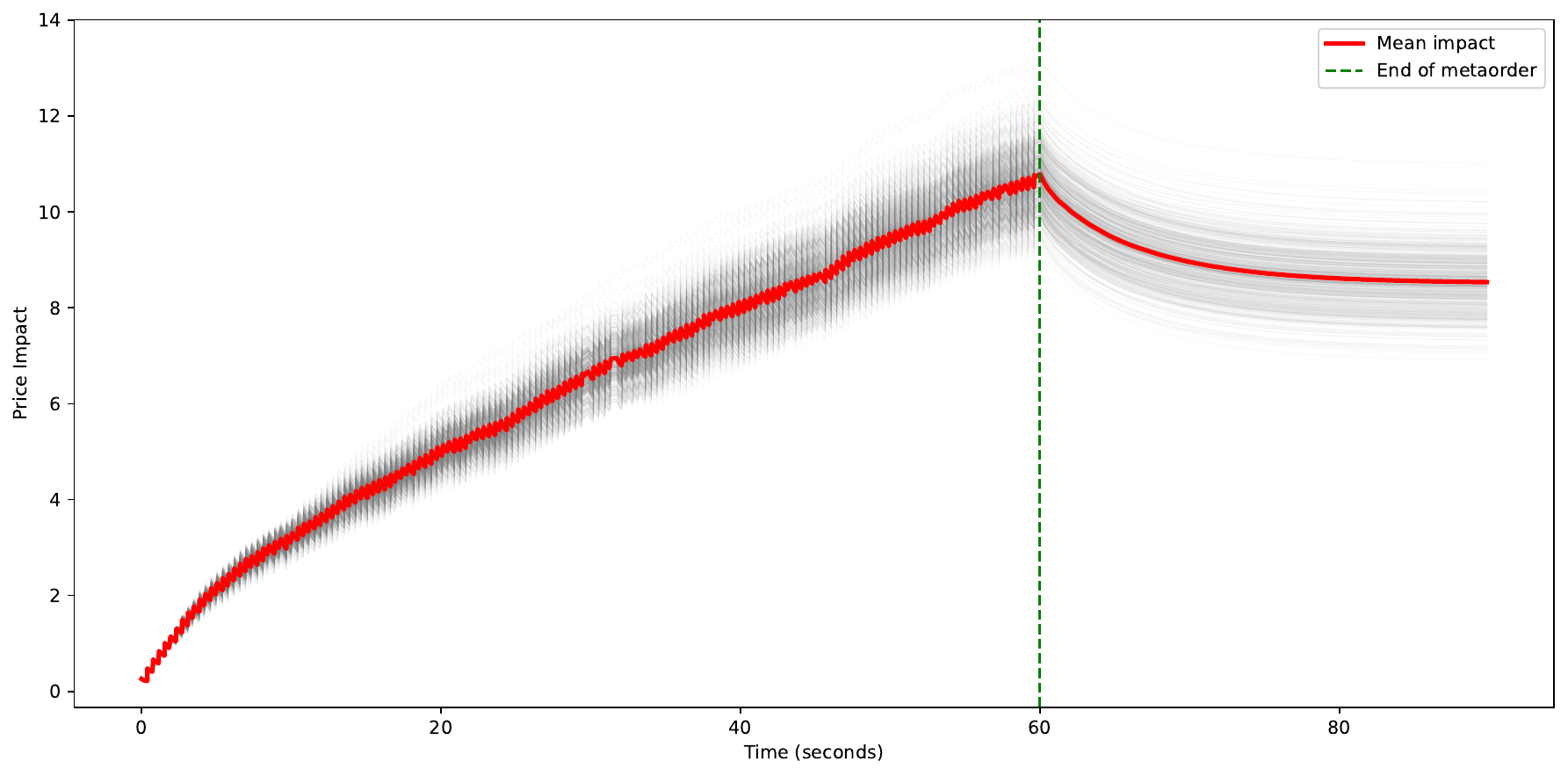}
    \caption{Conditional simulation of market impact given the observed baseline $q$ in the presence of a market metaorder.}
    \label{fig:mi_agressive}
\end{figure}

The model also exhibits permanent impact: as $t\to\infty$, the second integral vanishes as the transient kernel $\xi$ decays and the perturbed queue returns to its stationary regime, while the first integral freezes once the two queues merge. The permanent component is therefore the cumulative cost accrued through market orders before $\tau + 1$ minutes for some $\tau > 0$ over which $\wb q^a$ and $q^a$ remain distant; once $\kappa(\wb q^a_s)=\kappa(q^a_s)$, no further impact accumulates.

\begin{remark}
In this section, we illustrated our methodology using Hawkes market order processes taken independent of the remaining book dynamics, which yields closed-form expressions for the impact. More generally, the market order flow can be given a state dependent intensity $\lambda^M$, for instance, one coupled to the queue itself, but the impact in Equation~\eqref{eq:mi} must then be evaluated by nested Monte Carlo rather than in closed form. We further restricted the illustration to metaorders executed exclusively through either market or limit orders. Mixed strategies are of course also possible; from a practical standpoint, however, they would require a dedicated model for execution probabilities and order patience, an engineering task that lies beyond the scope of this paper.
\end{remark}

\FloatBarrier

\section{A posteriori evaluation of market impact}
\label{sec:a_posteriori}

Section~\ref{sec:quantifying} solved the forward problem: starting from baseline dynamics, one inserts an intervention and quantifies the induced distortion. Here we study the inverse post-trade problem, which is the relevant object for ex post impact and cost analysis \cite{jaisson2015market,youssef2024passiveimpact}: the strategy has already been executed, the impacted trajectory is observed, and the missing quantity is the baseline trajectory that would have prevailed without intervention under the same realized exogenous noise. Thus, in this section we attempt to answer the following question:

\begin{center}
\textit{Given a past market realization in which we executed a trading strategy: what would the market have looked like  had we not been present, and what were our true impact  and execution costs?}
\end{center}
\vspace{4mm}
We keep the one-sided notation of Section~\ref{sec:lob_sim}. For a mixed strategy, $L^o$ denotes the passive leg and $N^o$ the aggressive leg. On $[0,T]$, the observed impacted queue is
\begin{equation*}
\wb q_t = q_0 + \wb L_t - \wb C_t - \wb N_t + L^o_t - N^o_t,
\end{equation*}
while the unknown baseline queue is
\begin{equation*}
q_t = q_0 + L_t - C_t - N_t.
\end{equation*}
Hence the objective is to sample the conditional baseline law given the realized impacted trajectory,
\begin{equation*}
\mathcal L\!\left((L,C,N,q)\,\middle|\,(\wb q_s,L^o_s,N^o_s)_{0\le s\le T}\right),
\end{equation*}
under the same thinning-based coupling as in Sections~\ref{sec:market_impact_formula} and~\ref{sec:quantifying}.

\subsection{A posteriori conditional reconstruction}

For each valuation time $t\in[0,T]$, define the non-anticipative truncation
\begin{equation*}
L^{o,t}_s := L^o_{s\wedge t},
\qquad
N^{o,t}_s := N^o_{s\wedge t},
\qquad s\in[0,T],
\end{equation*}
and denote by $\wb q^{\,t}$ the corresponding truncated impacted queue. This is the same information structure as in Section~\ref{sec:market_impact_formula}. We set the counterfactual filtration
\begin{equation*}
\wb{\mathcal F}_t
:=
\sigma\bigl((\wb L_s,\wb C_s,\wb N_s,\wb q_s,L^o_s,N^o_s),\,0\le s\le t\bigr).
\end{equation*}

\begin{algorithm}[htbp]
\caption{A posteriori conditional simulation of $(L,C,N,q)$ given $\wb{\mathcal F}_T$}
\label{algo:posteriori_conditional_simulation}
\KwIn{Observed impacted trajectory $(\wb L,\wb C,\wb N,\wb q,L^o,N^o)$ on $[0,T]$; jump times $(T_j^{\wb x})_{j\ge1}$ of each $\wb x\in\{\wb L,\wb C,\wb N\}$; intensities $\lambda^L,\lambda^C,\lambda^N$.}
\KwOut{One conditional draw of $(L,C,N,q)$ on $[0,T]$.}
Set $t\leftarrow 0$, $L_0\leftarrow 0$, $C_0\leftarrow 0$, $N_0\leftarrow 0$, and $q_0\leftarrow q_0$\;
For each $\wb x\in\{\wb L,\wb C,\wb N\}$ and each jump time $T_j^{\wb x}$, draw and store $U_j^{\wb x}\sim\mathcal U([0,1])$\;
Set $j_x\leftarrow1$ for each $x\in\{L,C,N\}$\;
\While{$t<T$}{
For each $x\in\{L,C,N\}$, draw $\tau^x\sim\mathrm{Exp}((\lambda^x(q_t)-\lambda^x(\wb q_t))_+)$, with $\tau^x=+\infty$ if the rate is zero\;
For each $x\in\{L,C,N\}$, let $\theta^x\leftarrow T_{j_x}^{\wb x}-t$ be the waiting time to the next unprocessed impacted jump of type $x$ (or $+\infty$ if none remains)\;
Set $\theta\leftarrow\min\{\theta^L,\theta^C,\theta^N\}$ and let $x^\star$ be the corresponding type\;
Set $\eta$ to the time from $t$ to the next jump of $L^o$ or $N^o$ (or $+\infty$ if no intervention jump remains before $T$)\;
Set $\tau\leftarrow \min\{\tau^L,\tau^C,\tau^N,\theta,\eta,T-t\}$ and $t\leftarrow t+\tau$\;
\If{$t=T$}{stop the loop\;}
\If{$\tau=\theta$}{
Let $j=j_{x^\star}$\;
\If{$U_j^{\wb x^\star}\lambda^{x^\star}(\wb q_{t-})\le \lambda^{x^\star}(q_{t-})$}{
Increase the corresponding counter $L$, $C$, or $N$ by one at time $t$\;
}
$j_{x^\star}\leftarrow j_{x^\star}+1$\;
}
\ElseIf{$\tau=\eta$}{
Advance through the known intervention jump; no baseline counter is changed\;
}
\Else{
Let $x$ be such that $\tau=\tau^x$ and increase the corresponding counter $L$, $C$, or $N$ by one at time $t$\;
}
Update $q_t$ from $q_t=q_0+L_t-C_t-N_t$\;
}
\Return $(L,C,N,q)$ on $[0,T]$\;
\end{algorithm}

\begin{proposition}
\label{prop:posteriori_conditional_simulation}

Assume the one-sided queue dynamics are driven by thinning on a common Poisson noise source as in Sections~\ref{sec:primer} and~\ref{sec:conditional_measures}. Then Algorithm~\ref{algo:posteriori_conditional_simulation} generates an exact sample from
\begin{equation*}
\mathcal L\big((L,C,N,q)\mid\wb{\mathcal F}_T\big).
\end{equation*}

\end{proposition}

\begin{remark}
    It is the reverse counterpart of Proposition~\ref{prop:conditional_simulation}: the roles of $(q,L,C,N)$ and  $(\wb q,\wb L,\wb C,\wb N)$ are exchanged, and intervention clocks are removed because $(L^o,N^o)$ are already observed in the conditioning $\sigma$-field.
\end{remark}

\begin{remark}
The practical interpretation is straightforward: once the impacted trajectory is fixed, all uncertainty on the baseline path comes from the unrevealed part of the latent Poisson noise characterized in Section~\ref{sec:conditional_measures}. The same argument extends to replacement tests. If strategy $A$ is observed and strategy $B$ is an alternative, one can reconstruct $q^B$ conditionally on $q^A$ under shared noise, either by composition ($q^A\to q$, then $q\to q^B$) or by a direct event-driven coupling. This yields a rigorous A/B ex post comparison where differences are attributable to execution rules rather than to independent resampling.
\end{remark}

\subsection{A posteriori impact and cost computation}

Let us start with an observed market realization where we actually intervened through market and limit orders $(\wb q,\wb L,\wb C,\wb N)$. Conditionally on $\wb{\mathcal F}_T$, generate i.i.d.\ baseline replicas
\begin{equation*}
\bigl(q^{(m)},L^{(m)},C^{(m)},N^{(m)}\bigr)_{1\le m\le M}
\sim
\mathcal{L}\!\left((q,L,C,N)\,\middle|\,\wb{\mathcal F}_T\right)
\end{equation*}
for some $M > 1$. For each replica, define the posteriori impact sample at time $t$ by
\begin{equation*}
\mathrm{MI}_t^{(m)}:=\wb P_t-P_t^{(m)},
\end{equation*}
where $\wb P_t$ is the observed impacted price and $P_t^{(m)}$ is the baseline price obtained by applying the same pricing rule as in Section~\ref{sec:market_impact_formula}. In particular, for passive and aggressive specifications, this corresponds to formulas of Equations  \eqref{eq:mi_passive_conditional} and \eqref{eq:mi_active_conditional} with baseline objects replaced by replica $m$.\\

The empirical distribution of $\{\mathrm{MI}_t^{(m)}\}_{m=1}^M$ approximates the conditional law of realized impact, and
\begin{equation*}
\widehat{\mathrm{MI}}_t^{\,\mathrm{post}}:=\frac{1}{M}\sum_{m=1}^M\mathrm{MI}_t^{(m)}
\end{equation*}
provides a Monte Carlo estimator of the conditional mean. This separates the statistical problem into two steps: Algorithm~\ref{algo:posteriori_conditional_simulation} reconstructs the latent baseline order flow compatible with the observed impacted path, and the pricing functional is then evaluated on each reconstructed path. The second step can use either the full conditional formula of Equations \eqref{eq:mi_passive_conditional}--\eqref{eq:mi_active_conditional} or the closed-form approximation in Equation \eqref{eq:closed_form_mi_passive}, depending on the specification retained for the numerical experiment.\\

For cost analysis of a mixed strategy, let $L^{o,\mathrm{exe}}$ denote the executed passive volume (with $\dd L_t^{o,\mathrm{exe}}\le \dd L_t^o$). For any price path $P$, define
\begin{equation*}
\mathcal{C}_T(P)
:=
\int_0^T P_{t-}\,\dd L_t^{o,\mathrm{exe}}
\;+\;
\int_0^T P_{t-}\,\dd N_t^o.
\end{equation*}
Then, for each replica,
\begin{equation}
\label{eq:cost_function}
\Delta \mathcal{C}_T^{(m)}
:=
\mathcal{C}_T(\wb P)-\mathcal{C}_T(P^{(m)})
=
\int_0^T (\wb P_{t-}-P^{(m)}_{t-})\,\dd L_t^{o,\mathrm{exe}}
\;+\;
\int_0^T (\wb P_{t-}-P^{(m)}_{t-})\,\dd N_t^o,
\end{equation}
which separates passive and aggressive contributions and matches implementation-shortfall decompositions \cite{almgren2001optimal}.

\subsection{Numerical simulations}

While the results of this section extend straightforwardly to mixed strategies, we focus on strategies executed exclusively through either limit or market orders. We use the same parameters for the order-arrival intensities, the impact function, and the trading horizon as in Section~\ref{sec:quantifying}. The left panel of Figure~\ref{fig:passive_simulation_aposterio} shows the factual queue $\bar q$ (black), observed under the executed metaorder, together with the simulated counterfactual baseline trajectories $q$ (grey) and their average (red), for a passive metaorder accounting for 10\% of the average limit order flow. The right panel shows the corresponding passive impact distribution. Figure~\ref{fig:aggressive_simulation_aposterio} shows the analogous results for an aggressive metaorder. The scripts used to reproduce these numerical experiments are available at \githubrepo.

\begin{figure}[ht]
    \centering

    \begin{subfigure}{0.48\textwidth}
        \centering
        \includegraphics[width=\linewidth]{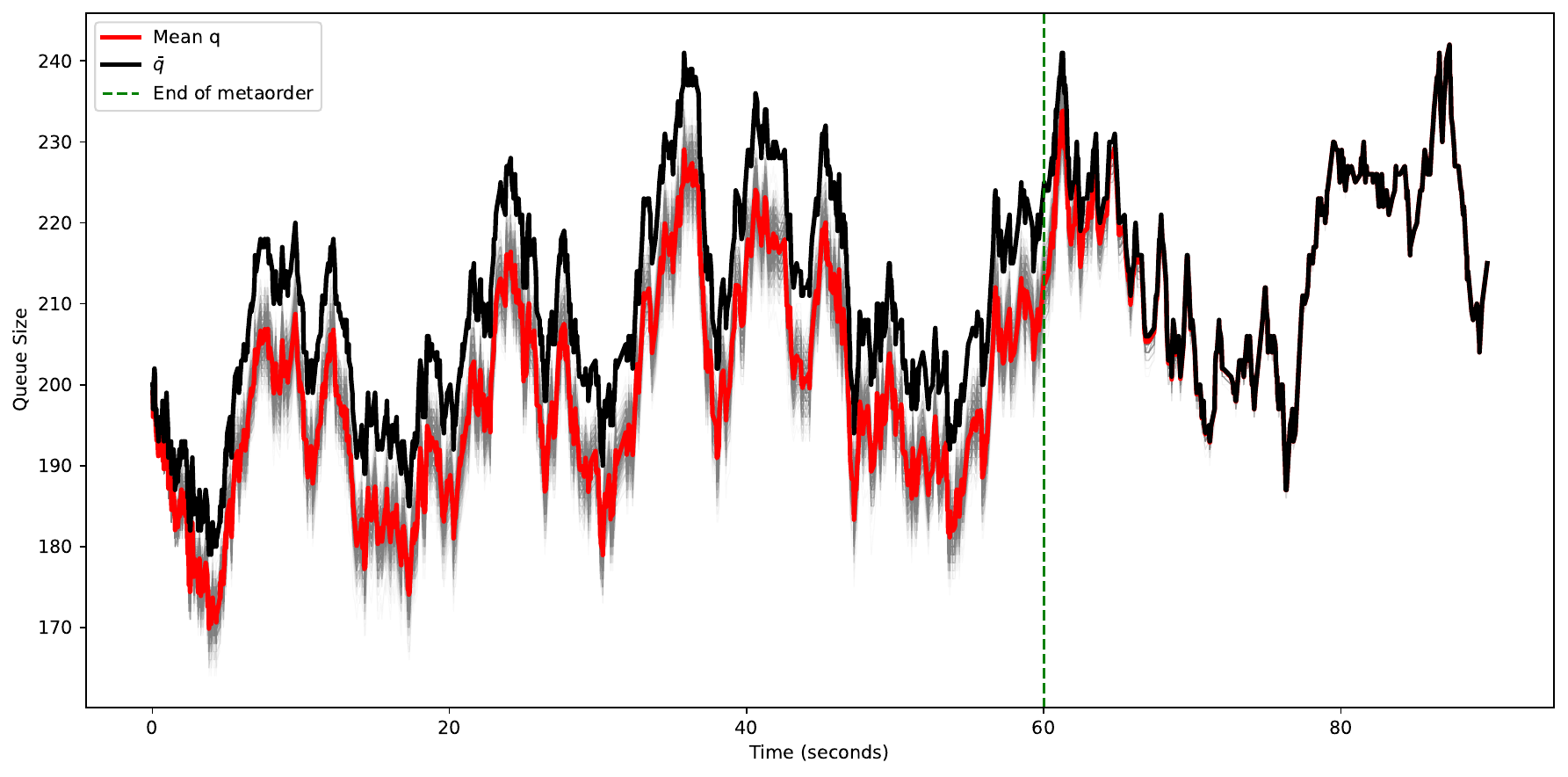}
        \caption{Conditional simulation of baseline $q$.}
    \end{subfigure}
    \hfill
    \begin{subfigure}{0.48\textwidth}
        \centering
        \includegraphics[width=\linewidth]{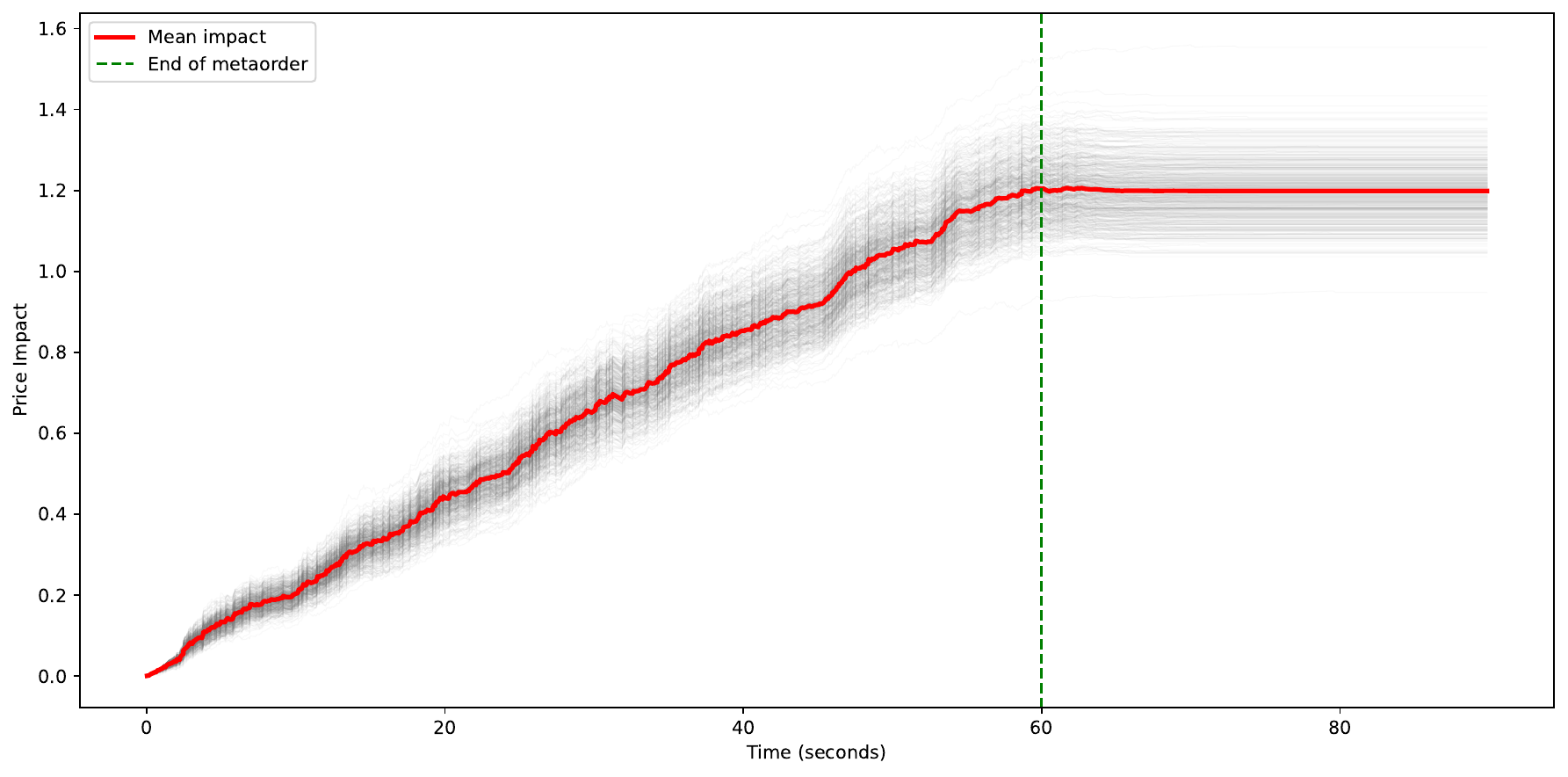}
        \caption{Conditional simulation of market impact.}
    \end{subfigure}

    \caption{Conditional simulation of baseline $q$ and the corresponding market impact given the observed intervened queue $\wb q$ in the presence of a passive metaorder.}
    \label{fig:passive_simulation_aposterio}
\end{figure}

\begin{figure}[ht]
    \centering

    \begin{subfigure}{0.48\textwidth}
        \centering
        \includegraphics[width=\linewidth]{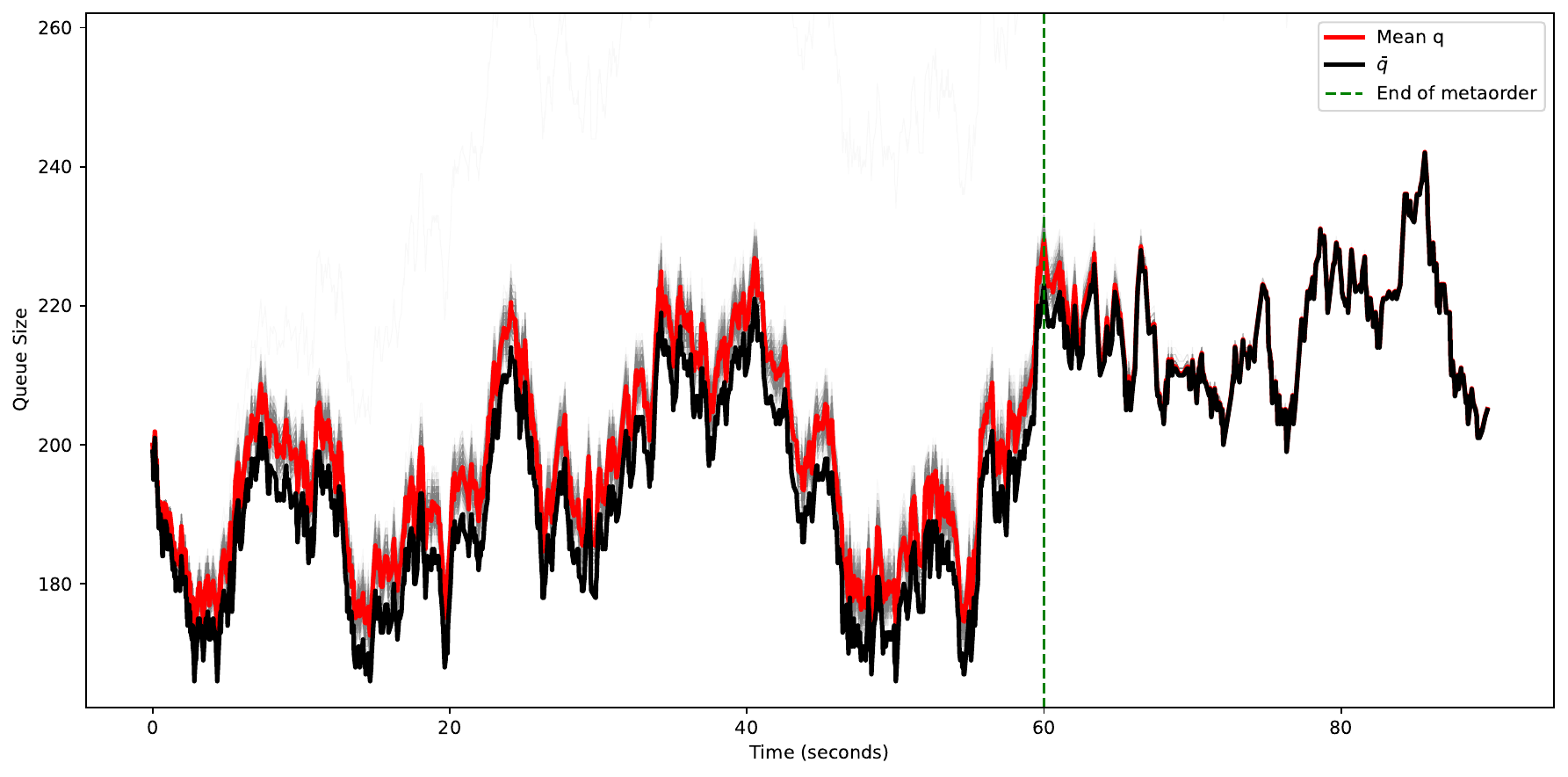}
        \caption{Conditional simulation of baseline $q$.}
    \end{subfigure}
    \hfill
    \begin{subfigure}{0.48\textwidth}
        \centering
        \includegraphics[width=\linewidth]{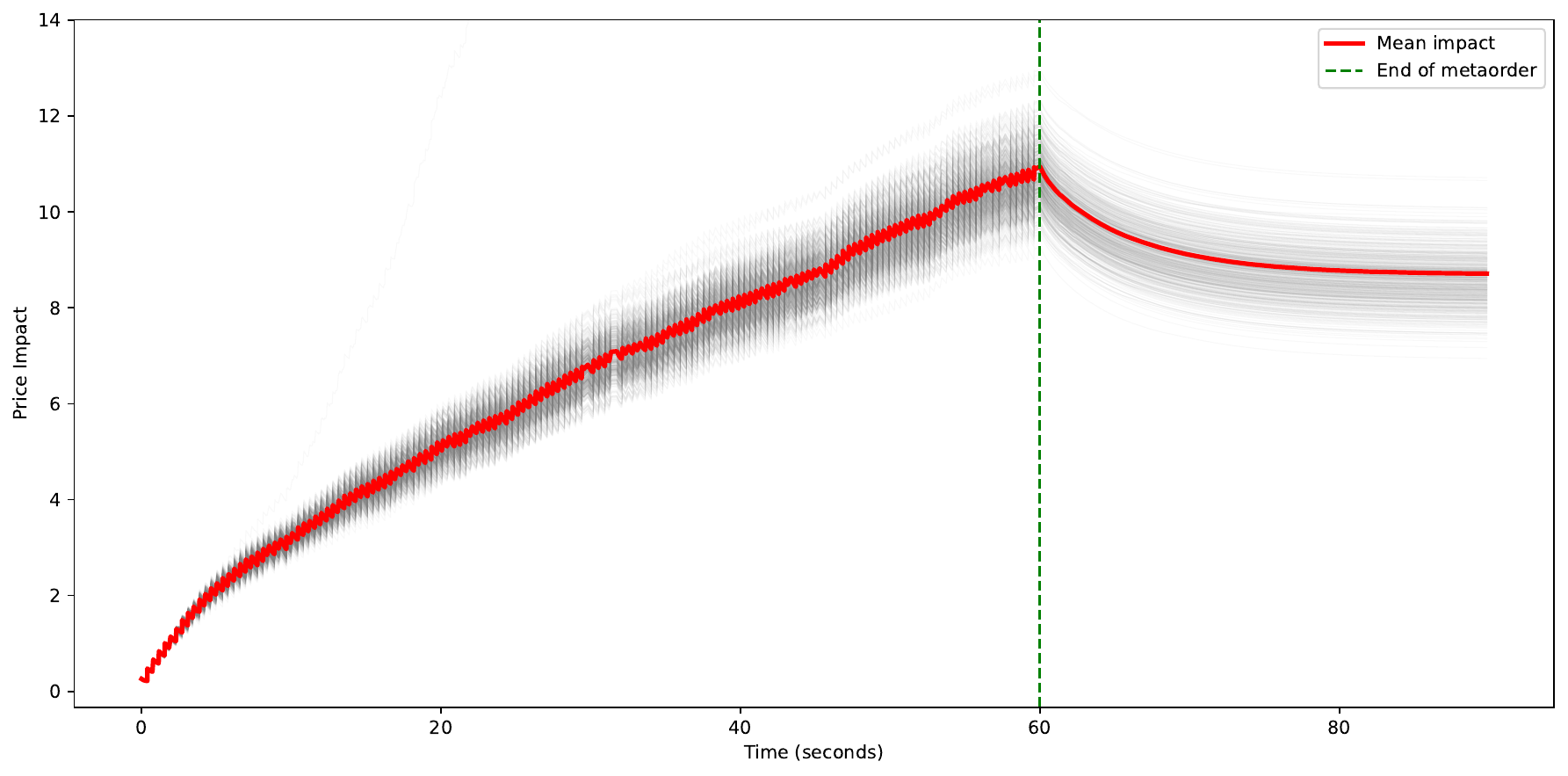}
        \caption{Conditional simulation of market impact.}
    \end{subfigure}

    \caption{Conditional simulation of baseline $q$ and the corresponding market impact given the observed intervened queue $\wb q$ in the presence of an aggressive metaorder.}
    \label{fig:aggressive_simulation_aposterio}
\end{figure}

\subsection{Real data application}

Let us now apply our methodology to a passive trading strategy on real data. We use
E-mini S\&P~500 futures order book data from May~29, 2025. The observed queue trajectory is treated as the intervened trajectory $\bar q$. Since our data are anonymous, we embed a deliberately simple strategy by flagging a subset of the observed limit and cancellation events, as follows. At initial time, we consider the next ten posted limit orders as part of our strategy and assume that, within a randomly sampled time $\tau$ \footnote{$\tau$ is exponentially distributed with mean such that, within time $\tau$, enough cancellations and markets are observed to carry on the strategy}, three of them are executed while the remaining seven are canceled. After time $\tau$, we wait enough time to observe 10 new posted limit orders, and the cycle is repeated three times. The rationale for this structure is that, by the time some of the orders are executed, the information on which they were originally posted has become stale, so the agents reposition themselves more appropriately.


We then proceed as follows:
\begin{itemize}
    \item First, calibrate the intensity functions $\lambda^L$ and $\lambda^C$, the Hawkes parameters $\mu$ and $\varphi$, and the impact function $\kappa$ from the observed trading history.
    entering Equation~\eqref{eq:mi_passive_conditional}.
    \item Conditionally on the observed intervened trajectory $\bar q$, we
    construct the no-strategy queue trajectory by removing the quantities
    flagged as belonging to the passive strategy.
    \item We estimate the passive market impact using
    Equation~\eqref{eq:closed_form_mi_passive}.
    \item Finally, we sample the market impact at the strategy's execution
    times and estimate the execution costs using
    Equation~\eqref{eq:cost_function}.
\end{itemize}

We use the following fitted parameters from \cite{hafsi2026impact} to evaluate
Equation~\eqref{eq:mi_passive_conditional}:
\begin{itemize}
    \item $c_\kappa = -0.00001713$,
    \item $\zeta = 0$,
    \item $\widetilde{\xi}(t)
    = -5.23666801e^{-10t}
    + 4.33922263e^{-t}
    + 8.26171935e^{-0.1t}
    + 1.63572603e^{-0.01t}$.
\end{itemize}

Figure~\ref{fig:mi_real_data} illustrates the result for this strategy. The
impact, shown in grey, rises when the flagged limit orders are posted and
decreases when the flagged cancellations occur, since cancellations have the
opposite sign to limit orders. The three cycles generate a sequence of impact
build-ups and reversals, while the resulting execution costs, obtained by
sampling the left-limit impact at execution times, are shown in blue. We obtain a peak price impact for the strategy of order a third of the tick size. This shows that even with a limited turnover, a passive strategy can generate non negligible impact.
\begin{figure}[!htbp]
    \centering
    \includegraphics[width=0.80\textwidth]{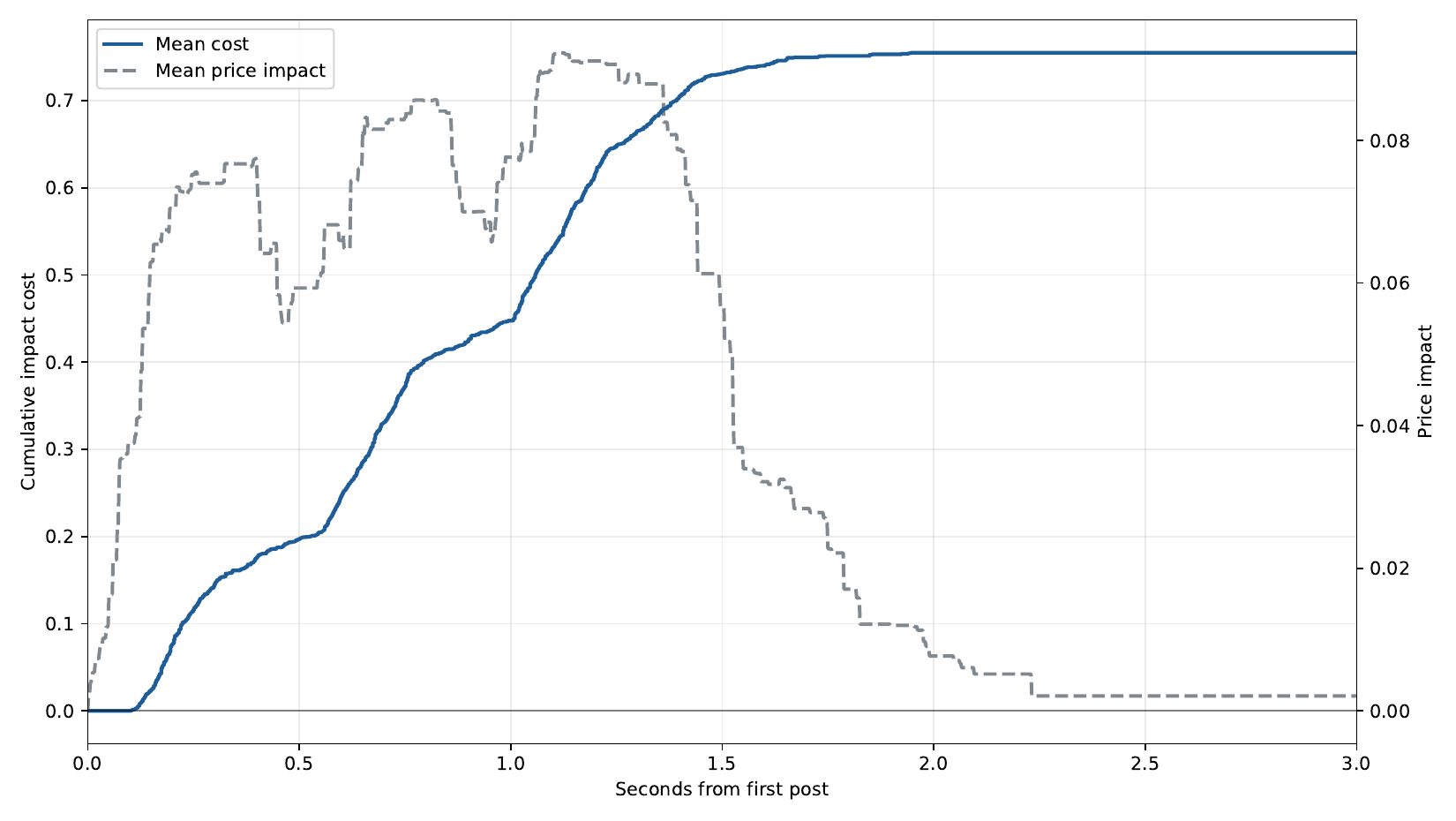}
    \caption{Conditional estimation of market impact and execution costs given
    the observed intervened queue $\wb q$ in the presence of a passive
    metaorder on real data. The numbers are expressed in index point.}
    \label{fig:mi_real_data}
\end{figure}

\begin{remark}
    For illustration, we constructed a specific strategy from the observed order flow. A broker or professional can of course feed the model their own trades and recover their own impact.
\end{remark}

\subsection*{Aknowledgments}
The authors gratefully acknowledge support from the ILB Chair \textit{Artificial Intelligence and Quantitative Methods for Finance} at University Paris Dauphine-PSL. The authors also thank Fabrizio Lillo, Yadh Hafsi, and Jesse Davis for inspiring discussions on market impact.

\FloatBarrier

\bibliographystyle{plainnat}
\bibliography{references}

\appendix

\section{Proof of Proposition \ref{proposition:market_impact_functional}}
\label{sec:prop_functional}
\label{appendix:section_two_one}

Fix $t\ge 0$. For $T\ge t$, set
\begin{equation*}
A_t:=\int_0^t \kappa(q_s^a)\,\dd N_s^a-\int_0^t \kappa(q_s^b)\,\dd N_s^b,
\qquad
Z_{t,T}:=\int_t^T \kappa(q_s^a)\,\dd N_s^a-\int_t^T \kappa(q_s^b)\,\dd N_s^b.
\end{equation*}
Since counting processes have no fixed jump time, $\Delta N_t^a=\Delta N_t^b=0$ a.s. Hence, from Equation \eqref{eq:price},
\begin{equation}
\label{eq:Pt_decomposition_forward_term}
P_t
=
P_0+A_t+\lim_{T\to\infty}\E\!\left[Z_{t,T}\mid\calG_t\right]
\qquad\text{a.s.}
\end{equation}
Therefore it is enough to show that the future term depends on $\calG_t$ only through \\
$(q_t^a,q_t^b,S_t(N^a),S_t(N^b))$.

For $y\in\calD$ and $r\ge 0$, define
\begin{equation*}
\Gamma_t(y)(r):=\int_{[0,t)}\varphi(t+r-s)\,\dd y(s).
\end{equation*}
This map is measurable and depends on $y$ only through $S_t(y)$. For $x\in\{a,b\}$,
\begin{equation}
\label{eq:future_hawkes_intensity}
\lambda_{t+r}^x
=
\mu+\Gamma_t(N^x)(r)+\int_{(t,t+r)}\varphi(t+r-s)\,\dd N_s^x,
\qquad r\ge 0,
\end{equation}
so the pre-$t$ contribution to future Hawkes intensities is fully encoded by
$S_t(N^a),S_t(N^b)$. Let
\begin{equation*}
\widehat\pi_t:=
\bigl(\vartheta_t\pi^{L,a},\vartheta_t\pi^{C,a},\vartheta_t\pi^{N,a},
\vartheta_t\pi^{L,b},\vartheta_t\pi^{C,b},\vartheta_t\pi^{N,b}\bigr),
\quad
(\vartheta_t\pi)(B):=\pi\bigl(\{(t+s,z):(s,z)\in B\}\bigr).
\end{equation*}
By independent increments of Poisson point measures, $\widehat\pi_t$ is independent of
$\calG_t$ and has the law of a fresh copy $\widehat\pi$. Fix deterministic $(u,v,n^a,n^b)\in\bbR^2\times\calD^2$. Given $\widehat\pi$, define for
$x\in\{a,b\}$, with $(u_a,u_b):=(u,v)$,
\begin{align*}
\widehat N_r^x
&=
\int_0^r\!\int_0^\infty
\1_{\left\{\xi\le \mu+\Gamma_t(n^x)(s)+\int_0^{s-}\varphi(s-u)\,\dd\widehat N_u^x\right\}}
\widehat\pi^{N,x}(\dd s,\dd\xi),
\\
\widehat X_r^x
&=
\int_0^r\!\int_0^\infty
\1_{\{z\le\lambda^X(\widehat q_{s-}^x)\}}
\widehat\pi^{X,x}(\dd s,\dd z),
\qquad X\in\{L,C\},
\\
\widehat q_r^x
&=
 u_x+\widehat L_r^x-\widehat C_r^x-\widehat N_r^x.
\end{align*}
This is the original queue/Hawkes system restarted at time $0$ with initial queues $(u,v)$ and
Hawkes prehistory terms $\Gamma_t(n^a),\Gamma_t(n^b)$. By the standard Poissonian pathwise
construction (see, e.g., \cite[Chapter~II]{bremaud1981point}), it has a unique strong solution; hence there exists a measurable map
\begin{equation*}
\Theta_t:\bbR^2\times\calD^2\times\mathfrak M^6\to\calD^4,
\qquad
\Theta_t(u,v,n^a,n^b,\widehat\pi)=(\widehat q^a,\widehat q^b,\widehat N^a,\widehat N^b),
\end{equation*}
which is non-anticipative in $(n^a,n^b)$ through $(S_t(n^a),S_t(n^b))$. Applying the same equations to the actual process after time $t$, and using Equation
\eqref{eq:future_hawkes_intensity}, pathwise uniqueness yields
\begin{equation}
\label{eq:future_solution_map}
\bigl(q_{t+\cdot}^a,q_{t+\cdot}^b,N_{t+\cdot}^a-N_t^a,N_{t+\cdot}^b-N_t^b\bigr)
=
\Theta_t\bigl(q_t^a,q_t^b,N^a,N^b,\widehat\pi_t\bigr)
\qquad\text{a.s.}
\end{equation}

For $n\in\bbN$, define
\begin{equation*}
G_n(\widehat q^a,\widehat q^b,\widehat N^a,\widehat N^b)
:=
\int_0^n \kappa(\widehat q_s^a)\,\dd\widehat N_s^a
-\int_0^n \kappa(\widehat q_s^b)\,\dd\widehat N_s^b.
\end{equation*}
Because $\kappa$ is bounded and $\widehat N^a,\widehat N^b$ are counting paths, $G_n$ is finite and measurable. Set
\begin{equation*}
\mathcal R^{(n)}(t,u,v,n^a,n^b)
:=
\E\!\left[G_n\!\left(\Theta_t(u,v,n^a,n^b,\widehat\pi)\right)\right].
\end{equation*}
Then $\mathcal R^{(n)}$ is measurable, and non-anticipative in $(n^a,n^b)$. From Equation 
\eqref{eq:future_solution_map} and independence of $\widehat\pi_t$ from $\calG_t$,
\begin{equation*}
\E\!\left[Z_{t,t+n}\mid\calG_t\right]
=
\mathcal R^{(n)}\bigl(t,q_t^a,q_t^b,N^a,N^b\bigr)
\qquad\text{a.s.}
\end{equation*}

Define
\begin{equation*}
\mathcal R_t(u,v,n^a,n^b):=\limsup_{n\to\infty}\mathcal R^{(n)}(t,u,v,n^a,n^b).
\end{equation*}
As a $\limsup$ of measurable non-anticipative maps, $\mathcal R_t$ is measurable and non-anticipative.
Since the limit in Equation \eqref{eq:price} exists a.s., evaluating Equation\eqref{eq:Pt_decomposition_forward_term} along
$T=t+n$ gives
\begin{align*}
P_t
&=
P_0+A_t+\lim_{n\to\infty}\mathcal R^{(n)}\bigl(t,q_t^a,q_t^b,N^a,N^b\bigr)
\\
&=
P_0
+\int_0^t \kappa(q_s^a)\,\dd N_s^a
-\int_0^t \kappa(q_s^b)\,\dd N_s^b
+\mathcal R_t\bigl(q_t^a,q_t^b,N^a,N^b\bigr)
\qquad\text{a.s.,}
\end{align*}
which is Equation \eqref{eq:functional_price_decomposition}. This also proves the stated non-anticipativity
property of $\mathcal R_t$.

\section{Poisson random measures, stochastic intensities, and thinning}
\label{appendix:poisson:measures}

Let $(E,\mathcal E)$ be a measurable space with $\sigma$-finite measure $\nu$. We write $\mathbb M(E)$ for the space of counting measures on $E$, equipped with the evaluation $\sigma$-field when measurability is needed. We use the following standard convention for Poisson random measures; see, for example, \cite{kingman1992poisson,daley2003introduction,last2017lectures}.

\begin{definition}
    \label{def:rpm}
    A Poisson random measure on $(E,\mathcal E)$ with intensity measure $\nu$ is a random measure $\pi:\Omega\to\mathbb M(E)$ such that:
    \begin{itemize}
        \item for every $\omega\in\Omega$, $\pi(\omega)$ is a counting measure on $(E,\mathcal E)$;
        \item if $C_1,\ldots,C_p\in\mathcal E$ are disjoint, then $\pi(C_1),\ldots,\pi(C_p)$ are independent;
        \item for every $C\in\mathcal E$ with $\nu(C)<\infty$, $\pi(C)$ has law $\Poiss(\nu(C))$.
    \end{itemize}
    We write $\Poisson(\nu)$ for its law.
\end{definition}

The following finite-intensity construction is the elementary scattering representation used in the simulation arguments; see \cite[Proposition~2.1.6]{RPG}.

\begin{proposition}[Construction of a Poisson measure with finite intensity]
    Let $\nu$ be a finite measure on $(E,\mathcal E)$. Let $M\sim\Poiss(\nu(E))$, and, on $\{\nu(E)>0\}$, let $(X_i)_{i\ge1}$ be i.i.d. with law $\nu(E)^{-1}\nu$, independent of $M$; when $\nu(E)=0$, set $M=0$. Then
    \begin{equation*}
    \Phi:=\sum_{i=1}^{M}\delta_{X_i}
    \end{equation*}
    is a Poisson random measure with intensity $\nu$.

    The construction extends to $\nu$ $\sigma$-finite with $E = \bigcup_{n \in \mathbb{N}} E_n$ for disjoint $E_n$ and $\nu(E_n)<+\infty$. One constructs $\Phi_n$ on each $E_n$ as previously explained as independent Poisson measures with intensity $\nu(-\cap E_n)$, and then  sets $\Phi := \sum_{n \in \mathbb{N}} \Phi_n$.
\end{proposition}

\begin{example}
\label{ex:simulation_method}
    A canonical example is a Poisson random measure on $\R_+^2$ with Lebesgue intensity. On the rectangle $[0,T]\times[0,M]$, its support can be generated through either of the following equivalent procedures:
    \begin{itemize}
        \item Sample $N\sim\Poiss(TM)$, then sample $((T_i,Z_i))_{i=1}^N$ i.i.d. with law $\mathcal{U}([0,T]\times[0,M])$.
        \item Generate i.i.d. pairs $((\tau_i,Z_i))_{i\ge1}$ with $\tau_i\sim\mathrm{Exp}(M)$ and $Z_i\sim\mathcal{U}([0,M])$, set $T_i=\sum_{j=1}^i\tau_j$, and stop at the first index $N$ such that $T_N>T$. 
    \end{itemize}
    This method is illustrated in Figure \ref{fig:my_image}.

    \begin{figure}[!htbp]
        \centering
        \includegraphics[width=0.9\textwidth]{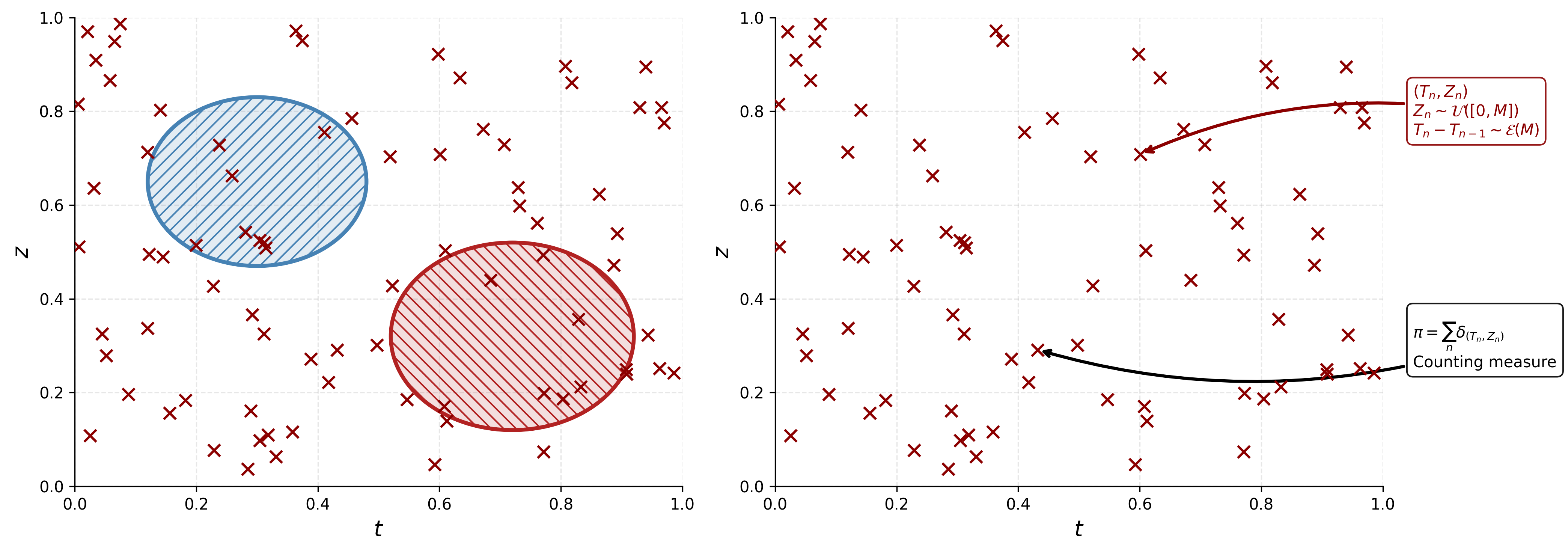}
        \caption{Poisson scattering on a rectangle and equivalent sampling descriptions.}
        \label{fig:my_image}
    \end{figure}
\end{example}

We now connect this geometric picture to counting processes with stochastic intensities.

\begin{definition}[Counting process]
    A counting process is an adapted process $N=(N_t)_{t\ge0}$ with values in $\mathbb N$, non-decreasing and right-continuous, such that $N_0=0$ and all jumps have size one. Unless stated otherwise, we work with non-explosive counting processes on finite horizons.
\end{definition}

The intensity is the predictable local rate of arrival. Informally,
\begin{equation*}
\mathbb P(N_{t+h}-N_t=1\mid\mathcal F_{t-})=\lambda_t h+o(h),
\qquad h\downarrow0,
\end{equation*}
with probability of two or more jumps of smaller order. The martingale-compensator definition is the form used below.

\begin{definition}[Counting process with stochastic intensity]
    Let $N=(N_t)_{t\ge0}$ be an $(\mathcal F_t)$-adapted counting process. We say that $N$ admits the $(\mathcal F_t)$-intensity $\lambda=(\lambda_t)_{t\ge0}$ if:
    \begin{enumerate}
        \item $\lambda$ is non-negative, locally integrable, and $(\mathcal F_t)$-predictable;
        \item $N_t-\int_0^t\lambda_s\,\dd s$ is an $(\mathcal F_t)$-local martingale.
    \end{enumerate}
\end{definition}

The constructive link with Poisson random measures is thinning \cite[Proposition~2.2.6]{RPG}.

\begin{proposition}[Thinning]
    Let $\Phi=\sum_{k\ge1}\delta_{X_k}$ be a Poisson random measure with intensity $\nu$ on $(E,\mathcal E)$. Let $(U_k)_{k\ge1}$ be i.i.d. uniform random variables on $[0,1]$, independent of $\Phi$, and let $p:E\to[0,1]$ be measurable. Then
    \begin{equation*}
    \sum_{k\ge1}\ind_{\{U_k\le p(X_k)\}}\delta_{X_k}
    \end{equation*}
    is a Poisson random measure with intensity $p(u)\nu(\dd u)$.
\end{proposition}

If $\pi$ is a Poisson random measure on $\R_+\times\R_+$ with Lebesgue intensity and $\lambda$ is predictable, then
\begin{equation*}
N_t=\int_0^t\int_{\R_+}\ind_{\{z\le\lambda_s\}}\,\pi(\dd s,\dd z)
\end{equation*}
has compensator $\int_0^t\lambda_s\,\dd s$ under the usual integrability assumptions; see \cite[Chapter~II]{bremaud1981point}. Equivalently, under a local deterministic upper bound $c$, proposal times are generated by a rate-$c$ clock and accepted with probability $\lambda/c$, which is the Lewis--Ogata simulation principle \cite{lewis1979simulation,ogata1981lewis}.

\begin{lemma}[Law of the next accepted jump]
\label{lem:next_jump}
    Let $N$ have predictable intensity $\lambda$ and first jump time
    \begin{equation*}
    \widetilde T:=\inf\{s>0:\Delta N_s=1\}.
    \end{equation*}
    Assume $\lambda_t\le c$ a.s. on $[0,\widetilde T]$ for some $c>0$. Let $(T_n)_{n\ge1}$ be the event times of a rate-$c$ Poisson process, and let $(U_n)_{n\ge1}$ be i.i.d. uniform variables on $[0,1]$, independent of that clock. Define
    \begin{equation*}
    S:=\inf\{n\ge1:\ U_n\le \lambda_{T_n}/c\}.
    \end{equation*}
    Then $T_S$, with $T_S=\infty$ on $\{S=\infty\}$, has the same law as $\widetilde T$.
\end{lemma}

\begin{example}[Hawkes process simulation]
    For a Hawkes process \cite{hawkes1971spectra} with intensity
    \begin{equation*}
    \lambda_t=\mu+\int_0^{t-}\varphi(t-s)\,\dd N_s,
    \end{equation*}
    where $\mu>0$ and non-negative $\varphi$, Ogata's algorithm updates the dominating bound after each accepted jump, so only a local bound is required between jumps \cite{ogata1981lewis,bacry2015hawkes}. 
    The last step of this process is illustrated on Figure~\ref{fig:hawkes_count}.
    \begin{figure}[!htbp]
        \centering
        \includegraphics[width=0.9\textwidth]{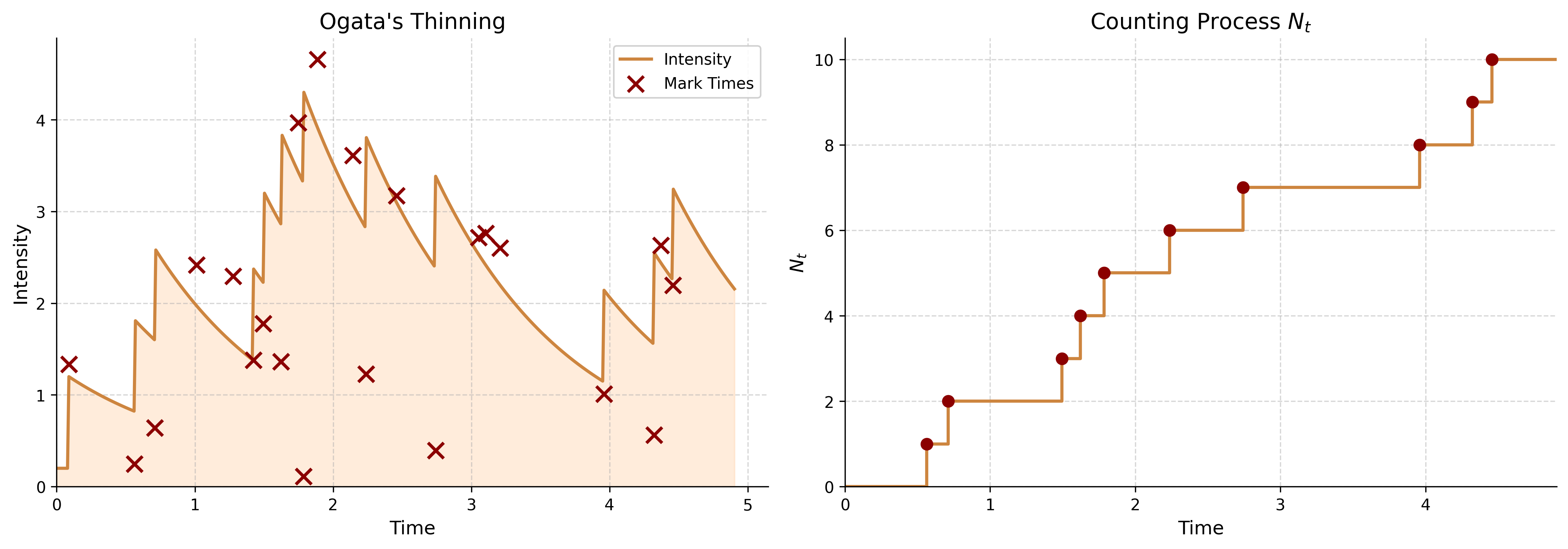}
        \caption{Thinning procedure for a Hawkes process.}
        \label{fig:hawkes_count}
    \end{figure}
\end{example}

\section{Proof of Theorem \ref{thm:conditionallaw}}
\label{sec:conditional_measures}
\label{appendix:section_three_tree}

Throughout this appendix, we work under the notation of Section~\ref{sec:cond_distribution}. We first recall some general definitions and results on random counting measures, random closed sets, and stopping sets, which are used in the proof of Theorem~\ref{thm:conditionallaw}. We then give a detailed structure of the proof, and finally prove the remaining lemmas.


\subsection{Stopping sets, measurability, and locality}

Throughout this subsection, we work on a general locally compact second-countable Hausdorff space $\mathbb G$ with Borel $\sigma$-field $\mathcal B(\mathbb G)$. 

\begin{definition}[Counting measures]
\label{def:counting_measure}
The space $\mathbb M(\mathbb G)$ consists of measures $\mu$ on $(\mathbb G,\mathcal B(\mathbb G))$ such that $\mu(C)\in\mathbb N$ for every relatively compact $C\in\mathcal B(\mathbb G)$. It is equipped with the $\sigma$-algebra
\begin{equation*}
\mathcal M(\mathbb G):=\sigma\bigl(\mu\mapsto\mu(C):\ C\in\mathcal B(\mathbb G)\bigr).
\end{equation*}
\end{definition}

Note that counting measures are locally finite. This structure is used repeatedly when applying measurable operations such as restriction, closure, and conditioning to Poisson configurations. For random closed sets, we use the Fell (hit-or-miss) topology.

\begin{definition}
Let $\mathcal C(\mathbb G)$ be the set of closed subsets of $\mathbb G$, equipped with the Fell topology, generated by
\begin{equation*}
\mathcal C^K_{O_1,\dots,O_n}
:=
\Bigl\{
F\in\mathcal C(\mathbb G):\ F\cap K=\varnothing,\ F\cap O_i\neq\varnothing,\ i=1,\dots,n
\Bigr\},
\end{equation*}
with $K\subset\mathbb G$ compact and $O_i\subset\mathbb G$ open. Let $\mathcal K(\mathbb G)\subset\mathcal C(\mathbb G)$ be the compact subsets.
\end{definition}

We refer to \cite[Definition 5.1.1]{TCCS} for a detailed presentation of the Fell topology. Note that the Borel $\sigma$-field on $\mathcal C(\mathbb G)$ is generated by hit sets
\begin{equation*}
\mathcal C_O:=\{F\in\mathcal C(\mathbb G):\ F\cap O\neq\varnothing\},
\qquad O\subset\mathbb G\ \text{open}.
\end{equation*}
Moreover, following \cite[Theorem 5.1.10.]{TCCS}; the Fell topology is exactly the topology induced by the Hausdorff distance when $\mathbb G$ is compact.

\begin{definition}[Random counting measure]
\label{def:random_counting_measure_appendix}
A random counting measure on $(\mathbb G,\mathcal B(\mathbb G))$ is a measurable map
\begin{equation*}
\pi:(\Omega,\mathcal F)\to(\mathbb M(\mathbb G),\mathcal M(\mathbb G)).
\end{equation*}
A Poisson random measure with intensity $\nu$ is a random counting measure whose finite-dimensional marginals satisfy Definition~\ref{def:rpm}.
\end{definition}

\begin{definition}[Stopping set]
Let $\pi$ be a random counting measure on $\mathbb G$, and let $S:\mathbb M(\mathbb G)\to\mathcal K(\mathbb G)$ be measurable. We call $S(\pi)$ a stopping set (with respect to $\pi$) if
\begin{equation*}
\{S(\pi)\subset K\}\in\sigma(\pi_{|K}),\qquad \forall K\in\mathcal K(\mathbb G).
\end{equation*}
\end{definition}

\begin{theorem}[Strong Markov property of Poisson point processes, {\cite[Theorem 12.1.3]{RPG}}]
\label{thm:strongmarkov}
Let $\pi$ be a Poisson random measure on $\mathbb G$, and let $S(\pi)$ be a stopping set with respect to $\pi$. Then, for every measurable map $F:\mathbb M(\mathbb G)\to\mathbb R_+$,
\begin{equation*}
\E[F(\pi)]
=
\E\!\left[F\!\left((\pi_{|S(\pi)})\cup(\pi'_{|S(\pi)^c})\right)\right],
\end{equation*}
where $\pi'$ is an independent copy of $\pi$.
\end{theorem}

\begin{definition}[Path spaces]
Let $\mathcal D$ be the space of càdlàg paths $\mathbb R_+\to\mathbb R^d$, and $\mathbb L$ the space of càglàd paths $\mathbb R_+\to\mathbb R^d$, both equipped with the Kolmogorov $\sigma$-field generated by coordinate maps.
\end{definition}

\begin{remark}[Measurability on càdlàg/càglàd spaces]
For $\mathcal D$ and $\mathbb L$, the Kolmogorov $\sigma$-field coincides with the Borel $\sigma$-field of the standard Skorokhod topologies; see \cite[Theorem 11.5.2]{SPL}.
\end{remark}

\subsection{Structure and completion of the proof}
\label{subsec:main_results}

The proof has three steps: 
\begin{itemize}
\item Step 1 proves a replacement identity in law: outside $B(\pi)$, $\pi$ can be replaced by an independent copy without changing expectations of functionals that keep track of the split $(\pi_{|B(\pi)},\pi_{|B(\pi)^c})$.
\item Step 2 converts this identity into conditional independence and identifies the conditional law on the complement as a Poisson random measure with restricted intensity $\ind_{B(\pi)^c}\nu$.
\item Step 3 upgrades the conditioning from $\sigma(B(\pi))$ to the observable $\sigma$-field $\mathcal F_T^{N^\pi}$, using that $B(\pi)$ is $\mathcal F_T^{N^\pi}$-measurable and that $N^\pi$ is determined by $\pi_{|B(\pi)}$.
\end{itemize}
This yields the first two bullets of Theorem~\ref{thm:conditionallaw}; the third one (conditional law of accepted marks) is proved afterwards. The proof of these points relies on the following properties of $A(\pi)$ and $B(\pi)$.

\begin{restatable}{lemma}{closuremeasurability}
\label{lem:closuremeasurability}
The map $\omega\mapsto B(\pi(\omega))$ is a measurable $\mathcal C(E)$-valued random compact set.
\end{restatable}

\begin{restatable}{lemma}{stability}
\label{lem:stability}
For all $\mu,\mu'\in\mathbb M(E)$,
\begin{equation*}
A(\mu)=A\bigl(\mu_{|B(\mu)}\cup\mu'_{|B(\mu)^c}\bigr),
\qquad
B(\mu)=B\bigl(\mu_{|B(\mu)}\cup\mu'_{|B(\mu)^c}\bigr).
\end{equation*}
\end{restatable}

\begin{restatable}{lemma}{stoppingset}
\label{lem:stoppingset}
The random set $B(\pi)$ is a stopping set with respect to $\pi$.
\end{restatable}

Lemmas~\ref{lem:closuremeasurability}--\ref{lem:stoppingset} are proved in Subsection~\ref{sec:conditional_measures:remaining_proofs}. Lemma~\ref{lem:closuremeasurability} ensures that conditioning with respect to $B(\pi)$ is well-posed and that $B(\pi)$ can be handled as a genuine random closed set. Lemma~\ref{lem:stability} is the locality statement: atoms outside $B(\mu)$ do not affect the accepted dynamics, hence neither $A(\mu)$ nor $B(\mu)$. Lemma~\ref{lem:stoppingset} is the bridge to Theorem~\ref{thm:strongmarkov}: once $B(\pi)$ is a stopping set, the strong Markov replacement can be applied on this random domain.

\begin{proof}[Proof of the conditional decomposition in Theorem~\ref{thm:conditionallaw}]
Let $\pi'$ be an independent copy of $\pi$, and set
\begin{equation*}
X:=\pi_{|B(\pi)},\qquad
Y:=\pi_{|B(\pi)^c},\qquad
Y':=\pi'_{|B(\pi)^c}.
\end{equation*}

\textbf{Step 1: replacement identity.}
For bounded measurable $h:\mathcal C(E)\to\mathbb R_+$, $f,g:\mathbb M(E)\to\mathbb R_+$, define
\begin{equation*}
\Phi_h(\mu):=h(B(\mu))\,f(\mu_{|B(\mu)})\,g(\mu_{|B(\mu)^c}).
\end{equation*}
Applying Theorem~\ref{thm:strongmarkov} to $\Phi_h(\mu)$ gives
\begin{equation*}
\E\!\left[h(B(\pi))f(X)g(Y)\right]
=
\E\!\left[
\Phi_h\!\left(\pi_{|B(\pi)}\cup\pi'_{|B(\pi)^c}\right)
\right].
\end{equation*}
Set $\widetilde\pi:=\pi_{|B(\pi)}\cup\pi'_{|B(\pi)^c}$.
Then
\begin{align*}
    \Phi_h(\widetilde\pi)
&=
h(B(\widetilde\pi))\,
f(\widetilde\pi_{|B(\widetilde\pi)})\,
g(\widetilde\pi_{|B(\widetilde\pi)^c}).
\end{align*}
Lemma~\ref{lem:stability}, applied with $(\mu,\mu')=(\pi,\pi')$, gives
$B(\widetilde\pi)=B(\pi)$. Therefore
\begin{align*}
\widetilde\pi_{|B(\widetilde\pi)}
&=
\left(\pi_{|B(\pi)}\cup\pi'_{|B(\pi)^c}\right)_{|B(\pi)}
=
\pi_{|B(\pi)}
=X,
\\
\widetilde\pi_{|B(\widetilde\pi)^c}
&=
\left(\pi_{|B(\pi)}\cup\pi'_{|B(\pi)^c}\right)_{|B(\pi)^c}
=
\pi'_{|B(\pi)^c}
=Y'.
\end{align*}
Hence
\begin{equation}
\label{eq:step1_identity}
\E\!\left[h(B(\pi))f(X)g(Y)\right]
=
\E\!\left[h(B(\pi))f(X)g(Y')\right].
\end{equation}

\textbf{Step 2: conditioning on $B(\pi)$.}
Let $Z$ be bounded and $\sigma(B(\pi))$-measurable. By the characterization of the Fell Borel $\sigma$-field, $\sigma(B(\pi))$ is generated by hit events $\{B(\pi)\in\mathcal C_O\}$ with $O\subset E$ open. A monotone-class argument therefore reduces to $Z=\ind_{\{B(\pi)\in\mathcal C_O\}}$. For such $Z$, take $h(F)=\ind_{\{F\in\mathcal C_O\}}$ in \eqref{eq:step1_identity}; then
\begin{equation*}
\E[Zf(X)g(Y)]=\E[Zf(X)g(Y')].
\end{equation*}
Hence
\begin{equation}
\label{eq:replaceY_Yprime}
\E[f(X)g(Y)\mid \sigma(B(\pi))]
=
\E[f(X)g(Y')\mid \sigma(B(\pi))].
\end{equation}
Since $\pi'$ is independent of $\pi$, conditionally on $\sigma(B(\pi))$,
\begin{align*}
    \E[f(X)g(Y')\mid\sigma(B(\pi))]
&=
\E[f(X)\mid\sigma(B(\pi))]\,
\E[g(Y')\mid\sigma(B(\pi))]
\\
&=
\E[f(X)\mid\sigma(B(\pi))]\,
\E[g(Y)\mid\sigma(B(\pi))].
\end{align*}
using \eqref{eq:replaceY_Yprime} with $f\equiv1$. Hence $X$ and $Y$ are conditionally independent given $\sigma(B(\pi))$. Moreover, conditionally on $B(\pi)$, $Y'$ is Poisson with intensity $\ind_{B(\pi)^c}\nu$, so
\begin{equation*}
\mathcal L\bigl(Y\mid \sigma(B(\pi))\bigr)
=
\Poisson\bigl(\ind_{B(\pi)^c}\nu\bigr).
\end{equation*}

\textbf{Step 3: conditioning on $\mathcal F_T^{N^\pi}$.}
We have $\sigma(B(\pi))\subset\mathcal F_T^{N^\pi}$, because $\lambda^\pi$ is a measurable functional of $N^\pi$. Also,
\begin{equation*}
\mathcal F_T^{N^\pi}\subset\sigma(X),
\end{equation*}
since Lemma~\ref{lem:stability} with $(\mu,\mu')=(\pi,0)$ yields $N^\pi=N^X$ on $[0,T]$.

Fix bounded measurable $g:\mathbb M(E)\to\mathbb R$. From Step 2, $Y$ is conditionally independent of $\sigma(X)$ given $\sigma(B(\pi))$, hence
\begin{equation*}
\E[g(Y)\mid \sigma(X)]=\E[g(Y)\mid \sigma(B(\pi))].
\end{equation*}
Taking conditional expectation with respect to $\mathcal F_T^{N^\pi}\subset\sigma(X)$, and using $\sigma(B(\pi))\subset\mathcal F_T^{N^\pi}$,
\begin{equation*}
\E[g(Y)\mid \mathcal F_T^{N^\pi}]
=
\E[g(Y)\mid \sigma(B(\pi))].
\end{equation*}
Therefore
\begin{equation*}
\mathcal L\bigl(Y\mid\mathcal F_T^{N^\pi}\bigr)
=
\Poisson\bigl(\ind_{B(\pi)^c}\nu\bigr).
\end{equation*}

Now let $f,g$ be bounded measurable. Using $\mathcal F_T^{N^\pi}\subset\sigma(X)$, Step 2, and the previous identity,
\begin{align*}
\E[f(X)g(Y)\mid \mathcal F_T^{N^\pi}]
&=
\E\!\left[f(X)\E[g(Y)\mid \sigma(X)]\mid \mathcal F_T^{N^\pi}\right] 
=
\E[g(Y)\mid \mathcal F_T^{N^\pi}]
\E[f(X)\mid \mathcal F_T^{N^\pi}],
\end{align*}
which proves the conditional independence of $X$ and $Y$ given $\mathcal F_T^{N^\pi}$.
\end{proof}

\begin{proof}[Proof of the conditional law of marks in Theorem~\ref{thm:conditionallaw}]
For $n\ge0$, let $(T_n^\mu,k_n^\mu)$ be the  $n$-th jump times or components of $N^\mu$, and define
\begin{equation*}
A_n(\mu):=\{(s,k,z)\in E:\ z\le\lambda_s^{\mu,k},\ s\le T_n^\mu\},
\qquad
B_n(\mu):=\overline{A_n(\mu)}.
\end{equation*}
As in Lemmas~\ref{lem:closuremeasurability} and~\ref{lem:stoppingset}, $B_n(\pi)$ is a stopping set. Let $\pi'$ be an independent copy of $\pi$, and set
\begin{equation*}
\tilde\pi^{(n)}:=\pi_{|B_n(\pi)}\cup\pi'_{|B_n(\pi)^c},
\qquad
g_n(\mu):=\mu_{|B_n(\mu)}.
\end{equation*}
Applying Theorem~\ref{thm:strongmarkov} to $B_n(\pi)$, then locality (Lemma~\ref{lem:stability} at level $n$),
\begin{equation*}
\E[F(g_{n+1}(\pi))]
=
\E[F(g_{n+1}(\tilde\pi^{(n)}))]
\end{equation*}
for every nonnegative measurable $F$. Hence the first $n+1$ accepted atoms under $\pi$ and $\tilde\pi^{(n)}$ have the same law. Conditionally on $\mathcal F_{T_n^\pi}^{N^\pi}$, all post-$T_n^\pi$ accepted atoms of $\tilde\pi^{(n)}$ are generated from the independent cloud $\pi'$ with predictable intensities
\begin{equation*}
\lambda_{n,j}(u):=f_j\bigl(S_{T_n^\pi}(N^\pi),T_n^\pi+u\bigr),
\qquad
\Lambda_n(u):=\int_0^u\sum_{j=1}^d \lambda_{n,j}(r)\,\dd r.
\end{equation*}
Let $E_{n+1}\sim\mathrm{Exp}(1)$, $U_{n+1}\sim\mathcal U([0,1])$, independent of $\mathcal F_{T_n^\pi}^{N^\pi}$ and of $\pi_{|B_n(\pi)}$. Set
\begin{equation*}
\tau_{n+1}:=\Lambda_n^{-1}(E_{n+1}),\qquad
T_{n+1}^{\tilde\pi^{(n)}}:=T_n^\pi+\tau_{n+1},
\end{equation*}
and choose $k_{n+1}$ with
\begin{equation*}
\mathbb P\!\left(k_{n+1}=j\mid \mathcal F_{T_n^\pi}^{N^\pi},\tau_{n+1}\right)
=
\Big(
\sum_{\ell=1}^d\lambda_{n,\ell}(\tau_{n+1})
\Big)^{-1}
\lambda_{n,j}(\tau_{n+1}).
\end{equation*}
By the standard first-jump thinning representation (Lemma~\ref{lem:next_jump}, componentwise),
\begin{equation*}
\bigl(T_{n+1}^{\tilde\pi^{(n)}},k_{n+1}^{\tilde\pi^{(n)}},Z_{n+1}^{\tilde\pi^{(n)}}\bigr)
=
\left(
T_n^\pi+\tau_{n+1},
k_{n+1},
U_{n+1}\lambda_{n,k_{n+1}}(\tau_{n+1})
\right)
\ \Big|\ \mathcal F_{T_n^\pi}^{N^\pi}
\quad\text{in law, conditionally on }\mathcal F_{T_n^\pi}^{N^\pi}.
\end{equation*}

Let
\begin{equation*}
\mathcal G_{n+1}:=
\sigma\bigl((T_i^\pi,k_i^\pi)_{1\le i\le n+1}\bigr)\vee\mathcal F_{T_n^\pi}^{N^\pi}.
\end{equation*}
Conditionally on $\mathcal G_{n+1}$, the only remaining randomness in $Z_{n+1}^\pi$ is $U_{n+1}$, so
\begin{equation*}
\E\!\left[h(Z_{n+1}^\pi)\mid \mathcal G_{n+1}\right]
=
\Big( \lambda_{T_{n+1}^\pi}^{\pi,k_{n+1}^\pi} \Big)^{-1}
\int_0^{\lambda_{T_{n+1}^\pi}^{\pi,k_{n+1}^\pi}} h(z)\,\dd z
\end{equation*}
for bounded measurable $h$. Moreover $U_{n+1}$ is independent of previous accepted marks, so $Z_{n+1}^\pi$ is conditionally independent of $(Z_i^\pi)_{1\le i\le n}$ given $\mathcal G_{n+1}$.

Iterating this one-step factorization, for every $m\ge1$ and bounded measurable $h_1,\dots,h_m$,
\begin{equation*}
\E\!\left[\prod_{i=1}^m h_i(Z_i^\pi)\,\middle|\,\mathcal F_T^{N^\pi}\right]
=
\prod_{i=1}^m
\frac{1}{\lambda_{T_i^\pi}^{\pi,k_i^\pi}}
\int_0^{\lambda_{T_i^\pi}^{\pi,k_i^\pi}} h_i(z)\,\dd z.
\end{equation*}
This is the finite-dimensional characterization of
\begin{equation*}
\mathcal L\bigl((Z_n)_{n\in I}\mid\mathcal F_T^{N^\pi}\bigr)
=
\bigotimes_{n\in I}\mathcal U\bigl([0,\lambda_{T_n^\pi}^{\pi,k_n^\pi}]\bigr).
\end{equation*}
\end{proof}

\subsection{Proof of remaining lemmas}
\label{sec:conditional_measures:remaining_proofs}

\begin{proof}[Proof of Lemma \ref{lem:closuremeasurability}]
For $g\in\mathbb L$, define
\begin{equation*}
\Gamma(g):=\overline{\{(s,k,z)\in E:\ z\le g_k(s)\}}\in\mathcal C(E).
\end{equation*}
It is enough to check measurability of $g\mapsto\Gamma(g)$ on hit sets of a countable basis of $E$. For
\begin{equation*}
O=(a,b)\times\{k\}\times(c,+\infty),
\end{equation*}
we have
\begin{align*}
\{g:\Gamma(g)\cap O\neq\varnothing\}
&=
\{g:\exists t\in(a,b),\ g_k(t)>c\} \\
&=
\{g:\sup_{t\in(a,b)} g_k(t)>c\} \\
&=
\bigcup_{q\in(a,b)\cap\mathbb Q}\{g:\ g_k(q)>c\},
\end{align*}
where the last equality uses the càglàd regularity of $g_k$. Hence $g\mapsto\Gamma(g)$ is measurable for the Kolmogorov $\sigma$-field on $\mathbb L$. This also implies measurability for $B(\pi)$ which is obtained by composition of measurable maps
\begin{equation*}
    \begin{tikzcd}
    B(\pi) \colon (\Omega,\mathcal F) \arrow[r, "\pi"] & \mathbb{M}(E) \arrow[r, "\mu\mapsto \lambda^\mu"] & \mathbb{L} \arrow[r, "\Gamma"] & \mathcal{C}(E).
    \end{tikzcd}
\end{equation*}


For compactness, by the non-explosion condition of Section~\ref{sec:cond_distribution},
\begin{equation*}
A(\mu)\subset
[0,T]\times\{1,\dots,d\}\times
\left[0,\sup_{t\in[0,T]}\sup_{1\le i\le d}\lambda_t^{\mu,i}\right],
\end{equation*}
and the right-hand side is compact. Thus $B(\mu)=\overline{A(\mu)}$ is compact for every $\mu$.
\end{proof}

\begin{proof}[Proof of Lemma \ref{lem:stability}]
Fix $\mu,\mu'\in\mathbb M(E)$, and set
\begin{equation*}
\tilde\mu:=\mu_{|B(\mu)}\cup\mu'_{|B(\mu)^c}.
\end{equation*}
Choose a compact set $K\subset E$ such that $B(\mu)\cup B(\tilde\mu)\subset K$. Since $(\mu+\mu')(K)<\infty$, the set
\begin{equation*}
S:=\{-1\}\cup
\Bigl\{t\in[0,T]:\ (\mu+\mu')\bigl(K\cap(\{t\}\times\{1,\dots,d\}\times\R_+)\bigr)>0\Bigr\}
\end{equation*}
is finite. Extend $N^\mu,N^{\tilde\mu}$ to $(-\infty,0)$ by $0$, and extend $\lambda^\mu,\lambda^{\tilde\mu}$ arbitrarily on $(-\infty,0)$ (this only initializes the induction). We prove by induction over $S$ that, for each $t\in S$,
\begin{equation*}
N^\mu=N^{\tilde\mu}\ \text{on }(-\infty,t],
\qquad
A(\mu)=A(\tilde\mu)\ \text{on }[0,t].
\end{equation*}

At $t=-1$, the claim is immediate. Let $t_0<t$ be consecutive elements of $S$, and assume the claim holds at $t_0$. First note that there are no atoms of $\mu$ or $\tilde\mu$ in $K\cap\bigl((t_0,t)\times\{1,\dots,d\}\times\R_+\bigr)$ and therefore
            \begin{align*}
            N_s^\mu &= N_{t_0}^\mu + \int_{t_0+}^s \int_{\{1,\cdots, d\} \times \R_+} e_k \ind_{\lp z \leq \lambda_u^\mu\rp} \mu(\dd u, \dd k, \dd z) \\
            &= N_{t_0}^{\tilde{\mu}} + \int_{t_0+}^s \int_{\{ 1,\cdots, d \} \times \R_+} e_k \ind_{\lp z \leq \lambda_u^{\tilde{\mu}} \rp} \tilde{\mu} (\dd u, \dd k, \dd z) = N_s^{\tilde{\mu}}
            \end{align*}
which means that $N^\mu$ and $N^{\tilde\mu}$ coincide on $(-\infty,t)$. Since $f$ is non-anticipative (Section~\ref{sec:primer}), the two intensities coincide up to time $t$. Therefore
\begin{equation*}
A(\mu)\cap\bigl([0,t]\times\{1,\dots,d\}\times\R_+\bigr)
=
A(\tilde\mu)\cap\bigl([0,t]\times\{1,\dots,d\}\times\R_+\bigr).
\end{equation*}
For each $k\in\{1,\dots,d\}$,
\begin{align*}
\Delta N_t^{\tilde\mu,k}
&=
\tilde\mu\bigl(\{t\}\times\{k\}\times[0,\lambda_t^{\tilde\mu,k}]\bigr) \\
&=
\mu\bigl(\{t\}\times\{k\}\times[0,\lambda_t^{\mu,k}]\bigr)
+\mu'\!\left(\bigl(\{t\}\times\{k\}\times[0,\lambda_t^{\mu,k}]\bigr)\cap B(\mu)^c\right).
\end{align*}
Since $\{t\}\times\{k\}\times[0,\lambda_t^{\mu,k}]\subset A(\mu)\subset B(\mu)$, the second term is $0$, hence $\Delta N_t^{\tilde\mu,k}=\Delta N_t^{\mu,k}$. Therefore $N_t^{\tilde\mu}=N_t^\mu$, closing the induction. \\

After the largest time in $S$, both measures have no atoms in $K$. Because $A(\mu)\subset B(\mu)\subset K$ and $A(\tilde\mu)\subset B(\tilde\mu)\subset K$, no accepted jump can occur afterwards. Thus $N^\mu=N^{\tilde\mu}$ on $[0,T]$, so $\lambda^\mu=\lambda^{\tilde\mu}$, hence $A(\mu)=A(\tilde\mu)$. Taking closures yields $B(\mu)=B(\tilde\mu)$.
\end{proof}

\begin{proof}[Proof of Lemma \ref{lem:stoppingset}]
Let $L\subset E$ be compact. We claim
\begin{equation*}
\{m\in\mathbb M(E):\ B(m)\subset L\}
=
\{m\in\mathbb M(E):\ B(m_{|L})\subset L\}.
\end{equation*}
If $B(m)\subset L$, apply Lemma~\ref{lem:stability} with $(\mu,\mu')=(m,m_{|L})$:
\begin{equation*}
B(m)
=
B\bigl(m_{|B(m)}\cup m_{|L\cap B(m)^c}\bigr)
=
B(m_{|L}),
\end{equation*}
so $B(m_{|L})\subset L$. Conversely, if $B(m_{|L})\subset L$, apply Lemma~\ref{lem:stability} with $(\mu,\mu')=(m_{|L},m)$:
\begin{equation*}
B(m_{|L})
=
B\bigl((m_{|L})_{|B(m_{|L})}\cup m_{|B(m_{|L})^c}\bigr) =
B\bigl(m_{|B(m_{|L})}\cup m_{|B(m_{|L})^c}\bigr)
=
B(m),
\end{equation*}
hence $B(m)\subset L$. Therefore we get $\{B(\pi)\subset L\}
=
\{B(\pi_{|L})\subset L\}$. By Lemma~\ref{lem:closuremeasurability}, $m\mapsto B(m)$ is measurable, and $m\mapsto m_{|L}$ is $\mathcal M(E)$-measurable; thus $\{B(\pi_{|L})\subset L\}\in\sigma(\pi_{|L})$. This is exactly the stopping-set criterion.
\end{proof}

\section{Proof of Proposition \ref{prop:conditional_simulation}} \label{appendix:B}

We keep the one-sided notation of Subsection~\ref{subsec:conditional_sim} on a fixed horizon $T>0$, and formalize Algorithm~\ref{algo:conditional_simulation}. Proposition~\ref{prop:conditional_simulation} can be rephrased as follows.

\begin{proposition}
    \label{prop:conditional_simulation:rephrased}
    The conditional simulation of $(\widebar{L}, \widebar{C}, \widebar{N})$ given $\mathcal{F}_T$ is event-driven.
    Assume that the construction is at an inspection time $t$: all observed jumps up to $t$ have been processed, and $\widebar q_t$ is known. For $x\in\{L,C,N\}$, let
    \begin{equation*}
    j_x(t):=\inf\{j\ge1:\ T_j^x>t\},
    \qquad
    \theta^x:=T_{j_x(t)}^x-t,
    \qquad
    \theta:=\min_{x\in\{L,C,N\}}\theta^x,
    \end{equation*}
    with $\theta^x=+\infty$ when no observed jump of type $x$ remains before $T$.
    Draw i.i.d. uniforms $(U_j^x)\sim \mathcal U([0,1])$, independent across $x\in\{L,C,N\}$, and conditionally on $(\mathcal F_T,\widebar q_t)$ draw independent clocks
    \begin{equation*}
    \tau^x\sim \mathrm{Exp}\Big((\lambda^x(\widebar q_t)-\lambda^x(q_t))_+\Big),
    \qquad x\in\{L,C,N\},
    \end{equation*}
    with the convention $\tau^x=+\infty$ when the rate is zero. Set
    \begin{equation*}
    \tau^{L^o}:=\inf\{s>0:\ L^o_{t+s}=L^o_t+1\},
    \qquad
    \tau:=\min\{\tau^L,\tau^C,\tau^N,\tau^{L^o},\theta,T-t\}.
    \end{equation*}
    Then:
    \begin{itemize}
        \item if $\tau=\tau^x$ for some $x\in\{L,C,N\}$, add one counterfactual jump of type $x$ at time $t+\tau$;
        \item if $\tau=\tau^{L^o}$, only the intervention path jumps and the queue is updated through $L^o$;
        \item if $\tau=\theta$, let $x^\star$ be the type attaining $\theta$, $s=t+\theta=T_{j_{x^\star}(t)}^{x^\star}$, and keep the observed atom if
        \begin{equation*}
        U_{j_{x^\star}(t)}^{x^\star}\lambda^{x^\star}(q_{s-})
        \le
        \lambda^{x^\star}(\widebar q_{s-}).
        \end{equation*}
    \end{itemize}
    In every case, the algorithm advances to $t+\tau$, and event times are almost surely distinct.
\end{proposition}

\begin{proof}
For each $X\in\{L,C,N\}$, let $\pi^X$ be the Poisson random measure driving the factual process
\begin{equation*}
X_t=\int_0^t\!\!\int_{\R_+}\ind_{\{z\le\lambda^X(q_{s-})\}}\,\pi^X(\dd s,\dd z).
\end{equation*}
Define $A_X^\pi:=\{(s,z)\in[0,T]\times\R_+:\ z\le\lambda^X(q_{s-})\}$.
By Theorem~\ref{thm:conditionallaw}, applied componentwise, and by the conditional law of accepted marks,
\begin{equation*}
\pi^X_{|A_X^\pi}
\stackrel{\mathcal L}{=}
\sum_{j\ge1}\delta_{(T_j^X,U_j^X\lambda^X(q_{T_j^X-}))},
\qquad
\mathcal L\!\left(\pi^X_{|(A_X^\pi)^c}\mid\mathcal F_T\right)
=
\Poisson\!\left(\ind_{(A_X^\pi)^c}\dd s\,\dd z\right),
\end{equation*}
with conditional independence between the two restrictions. Hence, conditionally on $\mathcal F_T$, $\pi$ has the same law as $\pi^\sharp=(\pi^{\sharp,L},\pi^{\sharp,C},\pi^{\sharp,N})$, where
\begin{equation*}
\pi^{\sharp,X}
:=
\sum_{j\ge1}\delta_{(T_j^X,U_j^X\lambda^X(q_{T_j^X-}))}
+\wb\pi^X_{|(A_X^\pi)^c},
\end{equation*}
and $\wb\pi^L,\wb\pi^C,\wb\pi^N$ are independent Poisson random measures of intensity $\dd s\,\dd z$, independent of the uniforms. Therefore it is enough to construct $(\wb L,\wb C,\wb N,\wb q)$ from $\pi^\sharp$:
\begin{equation*}
\wb X_t=\int_0^t\!\!\int_{\R_+}\ind_{\{z\le\lambda^X(\wb q_{s-})\}}\,\pi^{\sharp,X}(\dd s,\dd z),
\qquad
\wb q_t=q_0+\wb L_t-\wb C_t-\wb N_t+L_t^o.
\end{equation*}
Splitting observed and residual atoms, for each $X\in\{L,C,N\}$,
\begin{equation*}
\wb X_t
=
\sum_{T_j^X\le t}
\ind_{\{U_j^X\lambda^X(q_{T_j^X-})\le\lambda^X(\wb q_{T_j^X-})\}}
+\int_0^t\!\!\int_{\R_+}
\ind_{\{\lambda^X(q_{s-})<z\le\lambda^X(\wb q_{s-})\}}
\wb\pi^X(\dd s,\dd z).
\end{equation*}

Fix an inspection time $t$. Let $\theta$ and $\tau^{L^o}$ be as in the statement. On $(t,t+\theta\wedge\tau^{L^o})$, the factual path is constant; before the first new event, $\wb q$ is also constant. Conditionally on $(\mathcal F_T,\wb q_t)$, for each $X\in\{L,C,N\}$ the residual integral above generates a Poisson clock with rate $(\lambda^X(\wb q_t)-\lambda^X(q_t))_+$, hence waiting time $\tau^X\sim\mathrm{Exp}((\lambda^X(\wb q_t)-\lambda^X(q_t))_+)$. The next observed candidate time is $\theta$, and the intervention candidate time is $\tau^{L^o}$. Therefore the next transition time is
\begin{equation*}
\tau=\min\{\tau^L,\tau^C,\tau^N,\tau^{L^o},\theta,T-t\}.
\end{equation*}
If $\tau=\tau^X$, one residual atom of type $X$ is added. If $\tau=\tau^{L^o}$, only $L^o$ jumps. If $\tau=\theta$, with $x^\star$ denoting the corresponding type and $s=t+\theta=T_{j_{x^\star}(t)}^{x^\star}$, the observed atom contributes if
\begin{equation*}
U_{j_{x^\star}(t)}^{x^\star}\lambda^{x^\star}(q_{s-})
\le
\lambda^{x^\star}(\wb q_{s-}),
\end{equation*}
which is exactly the acceptance test in the statement. Thus one-step transitions of the algorithm coincide with those of the Poisson construction driven by $\pi^\sharp$. Since ties have probability zero, induction over successive inspection times on $[0,T]$ yields
\begin{equation*}
\mathcal L\big((\wb L,\wb C,\wb N,\wb q)\mid\mathcal F_T\big),
\end{equation*}
proving Proposition~\ref{prop:conditional_simulation:rephrased}, and therefore Proposition~\ref{prop:conditional_simulation}.
\end{proof}

\section{Proof of Theorem~\ref{thm:explicit_impact_shape}}
\label{appendix:C}

Set $U_s^t:=\widebar q_s^{a,t}-q_s^a$. Under Assumption~\ref{assum:linear},
\begin{equation*}
\kappa(\widebar q_s^{a,t})-\kappa(q_s^a)=c_\kappa U_s^t,
\end{equation*}
hence \eqref{eq:mi} gives
\begin{equation*}
\mathrm{MI}_t
=
c_\kappa\int_0^t U_s^t\,\dd N_s^a
+c_\kappa\E\!\left[\int_{(t,\infty)}U_s^t\,\dd N_s^a\ \middle|\ \calG_t\right].
\end{equation*}
Under the shared-noise construction, $U^t$ is driven by the Poisson measures of
$(L^a,C^a,\widebar L^{a,t},\widebar C^{a,t},L^o)$, while $N^a$ is driven by an
independent Poisson measure. Therefore $[U^t,N^a]\equiv0$, so
\begin{equation*}
\int_{(t,\infty)}U_s^t\,\dd N_s^a
=
\int_{(t,\infty)}U_{s-}^t\,\dd N_s^a
\quad\text{a.s.}
\end{equation*}
Using the conditional compensator identity for counting processes
\citep[Chapter II]{bremaud1981point},
\begin{equation*}
\E\!\left[\int_{(t,\infty)}U_{s-}^t\,\dd N_s^a\ \middle|\ \calG_t\right]
=
\int_t^\infty \E\!\left[U_s^t\lambda_s^a\mid\calG_t\right]\dd s.
\end{equation*}
Thus, to close the formula it remains to identify
$\E[U_s^t\lambda_s^a\mid\calG_t]$.

\begin{lemma}
\label{lem:queue_independence}
For all $s>t$, we have
\begin{equation*}
\E\!\left[U_s^t\lambda_s^a\mid\calG_t\right]
=
\E\!\left[U_s^t\mid\calF_t\right]\,
\E\!\left[\lambda_s^a\mid\calG_t\right]
\quad \text{ and } \quad
\E\!\left[U_s^t\mid\calG_t\right]
=
U_t^t e^{c_\lambda(s-t)}.
\end{equation*}
\end{lemma}
Plugging this into the previous display yields
\begin{equation}
\label{eq:mi_intermediaire}
\mathrm{MI}_t
=
c_{\kappa}\int_0^t U_s^t\,\dd N_s^a
+
c_{\kappa}U_t^t
\int_t^\infty e^{c_{\lambda}(s-t)}\E\!\left[\lambda_s^a\mid\calG_t\right]\,\dd s.
\end{equation}

To turn \eqref{eq:mi_intermediaire} into the explicit affine formula of
Theorem~\ref{thm:explicit_impact_shape}, define
$f_t(s):=\E[\lambda_s^a\mid\calG_t]$, $s\ge t$, and
\begin{equation*}
R_t^i:=\int_0^t \alpha_i e^{-\beta_i(t-u)}\,\dd N_u^a,\qquad 1\le i\le m,
\end{equation*}
so that Hawkes dynamics give
\begin{equation*}
f_t(s)
=
\mu+\sum_{i=1}^m R_t^i e^{-\beta_i(s-t)}
+\int_t^s\left(\sum_{i=1}^m \alpha_i e^{-\beta_i(s-u)}\right)f_t(u)\,\dd u.
\end{equation*}
The next lemmas solve this Volterra equation and integrate it against
$e^{c_\lambda(s-t)}$.

\begin{lemma}
\label{lem:propagator_formula}
There exist positive constants $(\lambda_j,c_j)_{1\le j\le m}$ such that
\begin{equation*}
\Bigl(\delta_0-\sum_{i=1}^m \alpha_i e^{-\beta_i(\cdot)}\Bigr)^{-1}
=
\delta_0+\sum_{j=1}^m c_j e^{-\lambda_j(\cdot)}
\end{equation*}
as convolution kernels on $\R_+$.
\end{lemma}

\begin{lemma}
\label{lem:closed_impact_term}
Under \eqref{eq:kernel_sum_exps},
\begin{equation*}
\int_t^\infty e^{c_\lambda(s-t)}\E\!\left[\lambda_s^a\mid\calG_t\right]\,\dd s
=
\zeta+\sum_{i=1}^m \Theta_i R_t^i,
\end{equation*}
where
\begin{align*}
\Theta_i
&:=
\frac{1}{\beta_i-c_\lambda}
+ \sum_{j=1}^m \frac{c_j}{\lambda_j-\beta_i}
\left(
\frac{1}{\beta_i-c_\lambda}-\frac{1}{\lambda_j-c_\lambda}
\right),\\
\zeta
&:=
\mu \left(
-\frac{1}{c_\lambda}
+ \sum_{j=1}^m \frac{c_j}{\lambda_j}
\left(
\frac{1}{c_\lambda-\lambda_j}-\frac{1}{c_\lambda}
\right)
\right).
\end{align*}
\end{lemma}

\begin{proof}[Proof of Theorem~\ref{thm:explicit_impact_shape}]
Combining \eqref{eq:mi_intermediaire} and Lemma~\ref{lem:closed_impact_term},
\begin{equation*}
\mathrm{MI}_t
=
c_\kappa\int_0^t U_s^t\,\dd N_s^a
+c_\kappa U_t^t\left(\zeta+\sum_{i=1}^m \Theta_i R_t^i\right).
\end{equation*}
Since $R_t^i=\alpha_i\int_0^t e^{-\beta_i(t-u)}\,\dd N_u^a$, define
$\gamma_i:=\alpha_i\Theta_i$. Then
\begin{equation*}
\sum_{i=1}^m \Theta_iR_t^i
=
\int_0^t\sum_{i=1}^m \gamma_i e^{-\beta_i(t-u)}\,\dd N_u^a,
\end{equation*}
which is exactly the statement of Theorem~\ref{thm:explicit_impact_shape}.
\end{proof}

\begin{proof}[Proof of Lemma \ref{lem:queue_independence}]
Fix $s>t$ and write $U_u:=U_u^t$. Following
\citep{youssef2024passiveimpact}, one can rewrite
\begin{equation*}
U_s
=
U_t
-\int_{(t,s]}\int_{\R_+}\1_{\{z\le -c_\lambda U_{u-}\}}\tilde\pi(\dd u,\dd z),
\end{equation*}
with $\tilde\pi$ independent of the Poisson measure $\pi^N$ driving $N^a$:
\begin{equation*}
N_s^a
=
N_t^a+\int_{(t,s]}\int_{\R_+}\1_{\{z\le\lambda_u^a\}}\pi^N(\dd u,\dd z),
\qquad
\lambda_u^a=\mu+\int_0^{u-}\varphi(u-r)\,\dd N_r^a.
\end{equation*}
Conditionally on $\calG_t$, $U_s$ depends only on
$\tilde\pi_{|(t,s]\times\R_+}$, whereas $\lambda_s^a$ depends only on
$\pi^N_{|(t,s]\times\R_+}$. Hence
\begin{equation*}
\E\!\left[U_s\lambda_s^a\mid\calG_t\right]
=
\E\!\left[U_s\mid\calG_t\right]\E\!\left[\lambda_s^a\mid\calG_t\right].
\end{equation*}
Set $m_t(s):=\E[U_s\mid\calG_t]$. Taking $\calG_t$-conditional expectation in
the equation for $U_s$ and using the compensator, we get
\begin{equation*}
m_t(s)=U_t+c_\lambda\int_t^s m_t(u)\,\dd u,
\end{equation*}
thus $m_t(s)=U_t e^{c_\lambda(s-t)}=U_t^t e^{c_\lambda(s-t)}$.
\end{proof}

\begin{proof}[Proof of Lemma \ref{lem:propagator_formula}]
Let $k(u):=\sum_{i=1}^m\alpha_i e^{-\beta_i u}\1_{\{u\ge0\}}$. Grouping equal
decay rates, assume $0<\beta_1<\cdots<\beta_m$. The Laplace transform of the
resolvent is
\begin{equation*}
F(x):=\mathcal L\!\left((\delta_0-k)^{-1}\right)(x)
=\Big(1-\sum_{i=1}^m\frac{\alpha_i}{x+\beta_i}\Big)^{-1},\qquad x>0.
\end{equation*}
Set
\begin{equation*}
R(x):=1-\sum_{i=1}^m\frac{\alpha_i}{x+\beta_i},
\qquad
x\in\R\setminus\{-\beta_1,\dots,-\beta_m\}.
\end{equation*}
Since $R'(x) > 0$, $R$ is strictly increasing on each interval between consecutive poles, and
\begin{equation*}
\lim_{x\to-\beta_i^-}R(x)=+\infty,\qquad
\lim_{x\to-\beta_i^+}R(x)=-\infty,\qquad
R(0)=1-\sum_{i=1}^m\frac{\alpha_i}{\beta_i}>0.
\end{equation*}
Hence $R$ has exactly $m$ simple zeros $-\lambda_j$, with $\lambda_j>0$.
Therefore
\begin{equation*}
F(x)=1+\sum_{j=1}^m\frac{c_j}{x+\lambda_j},
\qquad
c_j=\frac{1}{R'(-\lambda_j)}
=\Big(\sum_{i=1}^m\frac{\alpha_i}{(-\lambda_j+\beta_i)^2}\Big)^{-1}>0.
\end{equation*}
Taking inverse Laplace transforms yields the claim.
\end{proof}

\begin{proof}[Proof of Lemma \ref{lem:closed_impact_term}]
Set $f_t(s):=\E[\lambda_s^a\mid\calG_t]$, $s\ge t$, and $f_t(s):=0$ for
$s<t$. Taking $\calG_t$-conditional expectation in Hawkes dynamics yields
\begin{equation*}
f_t(s)
=
\mu+\sum_{i=1}^m R_t^i e^{-\beta_i(s-t)}
+\int_t^s\left(\sum_{i=1}^m\alpha_i e^{-\beta_i(s-u)}\right)f_t(u)\,\dd u.
\end{equation*}
Applying Lemma~\ref{lem:propagator_formula}, for $s\ge t$,
\begin{align*}
f_t(s)
&=
\sum_{i=1}^m R_t^i\left(
e^{-\beta_i(s-t)}
+\sum_{j=1}^m c_j\int_t^s e^{-\lambda_j(s-u)}e^{-\beta_i(u-t)}\,\dd u
\right)+\mu\left(
1+\sum_{j=1}^m\frac{c_j}{\lambda_j}\bigl(1-e^{-\lambda_j(s-t)}\bigr)
\right).
\end{align*}
We conclude by computing explicitely the $\dd u$ integrals and then by integrating this against the exponential weight.
\end{proof}

\end{document}